\documentclass[a4paper,reqno]{amsart}
\usepackage[all]{xy}           %For commutative diagrams
\usepackage{amssymb}           %For double arrows
\usepackage{hyperref}
\usepackage{eucal}
\usepackage{graphicx}
\usepackage{epsfig}
\usepackage[usenames,dvipsnames]{color}
\usepackage{color}
 \usepackage{subfigure}

\numberwithin{equation}{section}

\newtheorem{definition}{Definition}[section]
\newtheorem{lemma}[definition]{Lemma}
\newtheorem{theorem}[definition]{Theorem}
\newtheorem{proposition}[definition]{Proposition}

\newtheorem{remarkth}[definition]{Remark}

\renewcommand{\emph}[1]{{\bfseries\itshape{#1}}}
\numberwithin{figure}{section}

\newcommand{\R}{\mathbb{R}}      %Numeros reales
\newcount\ancho \newcount\anchom \newcount\anchoa
\newcount\anchob \newcount\altura

\newcommand{\ltilde}[3][0]{\altura=0 \advance\altura by #1
           \ancho=#2 \anchom=\ancho \divide\anchom by 2
           \anchoa=\ancho \divide\anchoa by 4
           \anchob=\anchom \advance\anchob by \anchoa
           \kern-3pt \begin{array}[b]{c}
           \begin{picture}(1,1)(\anchom,-\altura)
        \qbezier(0,2)(\anchoa,5)(\anchom,2)
        \qbezier(\anchom,2)(\anchob,-1)(\ancho,4)
        \qbezier(0,2)(\anchoa,4.5)(\anchom,1.8)
        \qbezier(\anchom,1.8)(\anchob,-1.5)(\ancho,4)
       \end{picture} \\[-4pt]{#3}
                       \end{array} \kern-4pt    }

\newcommand{\lhat}[3][0]{\altura=0 \advance\altura by #1
           \ancho=#2 \anchom=\ancho \divide\anchom by 2
           \anchoa=\ancho \divide\anchoa by 4
           \anchob=\anchom \advance\anchob by \anchoa
           \kern-3pt \begin{array}[b]{c}
           \begin{picture}(1,1)(\anchom,-\altura)
        \qbezier(0,2)(\anchoa,4)(\anchom,6)
        \qbezier(\anchom,6)(\anchob,4)(\ancho,2)
        \qbezier(0,2)(\anchoa,3.8)(\anchom,5.6)
        \qbezier(\anchom,5.6)(\anchob,3.8)(\ancho,2)
       \end{picture} \\[-4pt] {#3}
                       \end{array} \kern-4pt    }

\makeatletter
\newcommand\prol{\@ifstar{\@proldf}{\@prolpf}}  %% if * dual else primal
\def\@prolpf{\@ifnextchar[{\@prolpf@wrt}{\@prolpf@}}
\def\@prolpf@wrt[#1]#2{\@ifnextchar[{\@prolpf@wrt@at{#1}{#2}}{\@prolpf@wrt@{#1}{#2}}}
\def\@prolpf@wrt@at#1#2[#3]{\prolsymbol^{#1}_{#3}#2}
\def\@prolpf@wrt@#1#2{\prolsymbol^{#1}#2}
\def\@prolpf@#1{\@ifnextchar[{\@prolpf@at{#1}}{\@prolpf@@{#1}}}
\def\@prolpf@at#1[#2]{\prolsymbol_{#2}#1}
\def\@prolpf@@#1{\prolsymbol#1}
\def\@proldf{\@ifnextchar[{\@proldf@wrt}{\@proldf@}}
\def\@proldf@wrt[#1]#2{\@ifnextchar[{\@proldf@wrt@at{#1}{#2}}{\@proldf@wrt@{#1}{#2}}}
\def\@proldf@wrt@at#1#2[#3]{\prolsymbol^{*#1}_{#3}#2}
\def\@proldf@wrt@#1#2{\prolsymbol^{*#1}#2}
\def\@proldf@#1{\@ifnextchar[{\@proldf@at{#1}}{\@proldf@@{#1}}}
\def\@proldf@at#1[#2]{\prolsymbol^*_{#2}#1}
\def\@proldf@@#1{\prolsymbol^*#1}
\def\prolsymbol{\mathcal{T}}
\makeatother

\newcommand{\Lag}{\mathcal{L}}

\begin{document}
{\Large

\title[ The dynamics of an articulated $n$-trailer  vehicle]{ The dynamics of an  articulated $n$-trailer  vehicle}

\author[A.\ Bravo-Doddoli]{Alejandro\ Bravo-Doddoli}
\address{Alejandro Bravo: Depto. de Matem\'aticas, Facultad de Ciencias, UNAM,
Circuito Exterior S/N, Ciudad Universitaria,  Mexico City,  04510, Mexico}
\email{Bravododdoli@ciencias.unam.mx}

\author[L. C.\ Garc\'{\i}a-Naranjo]{Luis C.\ Garc\'{\i}a-Naranjo}
\address{Luis C.\ Garc\'{\i}a-Naranjo:
Departamento de Matem\'aticas y Mec\'anica \\
IIMAS-UNAM \\
Apdo Postal 20-726,  Mexico City,  01000, Mexico}
\email{luis@mym.iimas.unam.mx}



\begin{abstract}
We derive the reduced equations of motion for an articulated $n$-trailer vehicle that moves under its own inertia
on the plane. We show that  the energy level surfaces in the reduced space are $(n+1)$-tori and we classify the 
equilibria within them, determining their stability. A thorough description of the dynamics is given in the case $n=1$.
\end{abstract}

\maketitle

\section{Introduction}

We consider the dynamics of an articulated $n$-trailer vehicle that moves under its own inertia. 
Such system consists of a leading car, or truck, that is pulling $n$ trailers,
like a luggage carrier in the airport. The leading car and the trailers form a convoy that is subjected to 
$(n+1)$-nonholonomic
constraints, one for each body.

This  system is a canonical example in  nonholonomic motion planning,  which is fundamental in robotics,
and has been extensively considered from the control perspective (see e.g.  \cite{Laumond1, Murray, Laumondbook} and the references therein).

The constraint distribution defined by the $n$-trailer system has also received interest in differential geometry.
 It is a Goursat distribution (see \cite{Jean}) and, as  shown in \cite{Mont-Z},  all possible Goursat germs
of corank $n+1$ are realized at its different points.

The dynamics of wheeled vehicles  moving on the plane has been considered in e.g.  \cite{Bolzern, Furta, Lutsenko}. The first
two of these deal with certain properties of the $n$-trailer system. However, the majority of the existing  references to 
this  system   deal with its kinematics  and  
 disregard its dynamical aspects. 
 To the best of our knowledge, a detailed study of the dynamics of the $n$-trailer system
is missing.

 The nonholonomic constraints in the $n$-trailer system
 arise by assuming that each of the bodies in the convoy has a pair of wheels that
prohibit  motion in the direction perpendicular to them.  Each of these constraints is
identical to the nonholonomic constraint for the well-known Chaplygin sleigh problem  \cite{Chap}. Hence, the  
articulated $n$-trailer vehicle is   a generalization of the Chaplygin sleigh system that is recovered when the number 
of trailers $n=0$. (Other generalizations of the Chaplygin sleigh are considered in \cite{Bo2009}).

We mention the thorough understanding of the dynamics of the Chaplygin sleigh has resulted in the design of control
algorithms, where the control mechanism moves the center of mass \cite{OZenk}. We are hopeful that the results in this paper will
turn out to be useful for control purposes.

The outline of the paper is as follows. In Section \ref{S:Def} we introduce the system and the notation that we will
follow in our work. We define the configuration space, the nonholonomic constraints and we
write down the kinetic energy of the system. We also discuss the related $SE(2)$ symmetry of the problem
and describe the  reduced space. In Section \ref{S:EqnsMain} we derive the reduced equations of motion \eqref{E:MainRed}. Our approach follows  the method suggested in \cite{Grab}. We
show that the energy level sets are diffeomorphic to $(n+1)$-tori, and we give working expressions for the restriction
of the flow to them. Section  \ref{S:dynagret0} considers the equilibria of the system assuming that
the center of mass of the leading car is displaced a distance $a>0$ from the wheel's axis. We give a complete classification
of all the equilibria in an energy level set and perform their  linear stability analysis. It is found that the straight line motion of the
convoy in the direction of the center of mass of the leading car and   with all of the trailers aligned behind it, is asymptotically stable.
In Section  \ref{S:dynaeq0} we deal with the case where the center of mass of the car lies on its wheels' axis. We give
necessary and sufficient conditions on the velocities for the existence of equilibria of the reduced system and we do an exhaustive treatment of the case
$n=1$. Finally, in Section \ref{S:sing} we comment on the interest to analyze the influence of the 
singular configurations on the dynamics.

\section{The $n$-trailer mobile vehicle}
\label{S:Def}

Following \cite{Laumond1, Murray} and other references given in these works, we consider a multi-body car system $(\mathcal{B}_0,\mathcal{B}_1, \dots , \mathcal{B}_n)$ that consists of a car $\mathcal{B}_0$ pulling $n$ trailers, $\mathcal{B}_1, \dots , \mathcal{B}_n$. The trailers form a 
convoy  (like in a luggage carrier)  that moves on the plane (see Figure \ref{F:Main-Diag} for the case $n=2$). 

Each body in the convoy has a set of wheels and we denote by $(x_i,y_i)$ the coordinates of the midpoint of the 
wheel's axis ($i=0,\dots, n$) with respect to a chosen cartesian frame.  The  orientation of $\mathcal{B}_i$ is
determined by the angle $\theta_i$ between the main axis of the body and the $x$ axis of the chosen frame (see Figure \ref{F:Main-Diag}).

\begin{figure}[h]
    \centering
    \includegraphics[width=16cm]{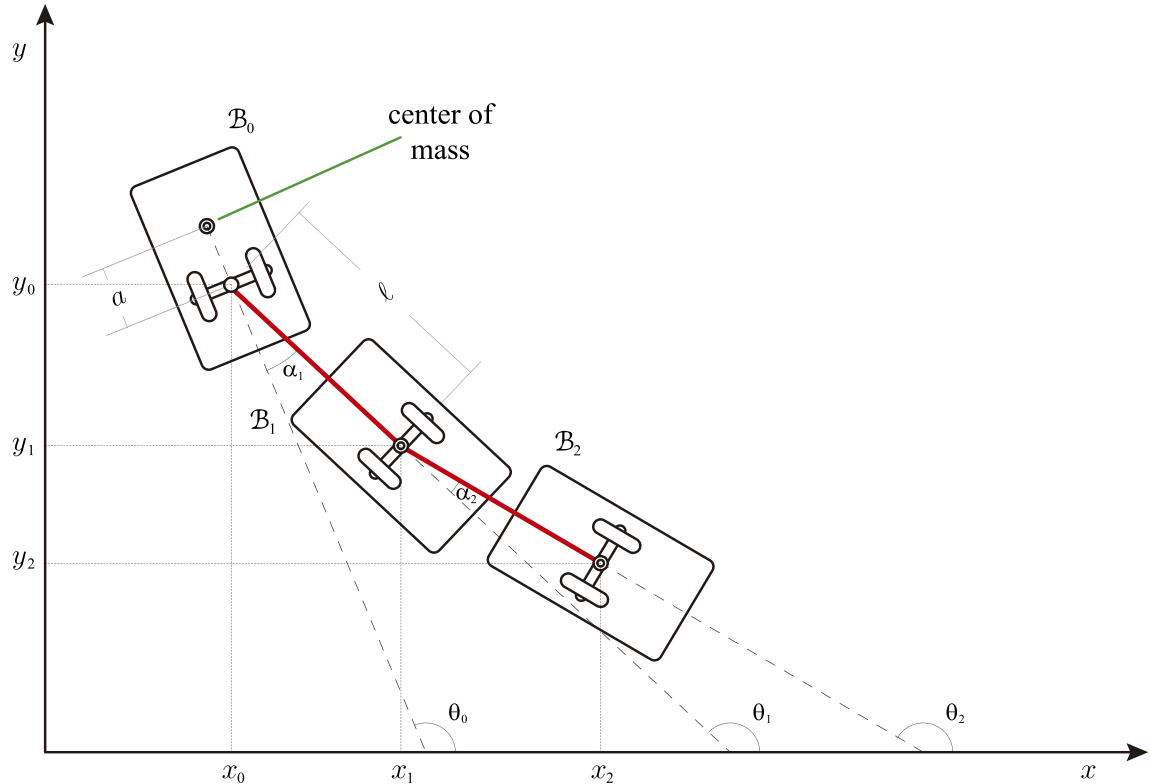}
    \caption{The $n$-trailer vehicle with $n=2$.}
    \label{F:Main-Diag}
\end{figure}

\subsection{Kinematics.} The convoy condition requires that the body $\mathcal{B}_i$ is hooked to the preceding body $\mathcal{B}_{i-1}$.
Following \cite{Laumond1, Murray} we assume that the hooking is done via a link of length $\ell$ that connects $(x_i,y_i)$ with
$(x_{i-1},y_{i-1})$ as illustrated in Figure \ref{F:Main-Diag}.\footnote{Other hooking mechanisms are possible
and have been considered in the literature. The Hilare robot at
LAAS Toulouse can realize various models, including the one that we consider in this paper \cite{Laumond2}.}  The hooking of the convoy thus defines the $2n$ holonomic constraints
\begin{equation}
\label{E:Hol-Const}
x_i+\ell \cos \theta_i -x_{i-1}=0, \qquad y_i+\ell \sin \theta_i -y_{i-1}=0, \qquad i=1,\dots, n.
\end{equation}

On the other hand, the wheels on each of the cars impose a nonholonomic constraint that forbids 
any motion of the given body in the direction perpendicular to its main axis. In this way we get the $n+1$ nonholonomic constraints
\begin{equation}
\label{E:NonHol-Const}
\dot x_i \sin \theta_i - \dot y_i \cos \theta_i =0, \qquad  i=0,1, \dots, n.
\end{equation}

In view of the holonomic constraints \eqref{E:Hol-Const}, the configuration of the convoy is fully determined by the 
value of the coordinates 
\begin{equation*}
x_0, y_0, \theta_0, \theta_1, \dots , \theta_n.
\end{equation*}
Therefore, the configuration space of the system is the $n+3$ dimensional manifold $Q=SE(2)\times \mathbb{T}^n$
where $SE(2)$ denotes the Euclidean group in the plane and $\mathbb{T}^n$ is the $n$-torus. The nonholonomic
constraints \eqref{E:NonHol-Const} define a rank $2$ constraint distribution $\mathcal{D}$ on $Q$.

\subsection{Dynamics.}

We assume that the center of mass of the leading car $\mathcal{B}_0$ is displaced a distance $a$ from the
midpoint of its wheel's axis $(x_0,y_0)$ along the principal axis of the body (see Figure \ref{F:Main-Diag}). Therefore, if $(x_C,y_C)$ denote the
coordinates of the center of mass of $\mathcal{B}_0$, we have 
\begin{equation}
\label{E:CofM}
x_C=x_0+a\cos \theta_0, \qquad y_C=y_0+a\sin \theta_0.
\end{equation}

We will denote the total mass of $\mathcal{B}_0$ by $M$ and  its moment of inertia about its center of mass by $J_0$. On the other hand, we shall suppose that the trailers $\mathcal{B}_1, \dots , \mathcal{B}_n$ are
identical, with their center of mass lying at the midpoint of the wheel's axis $(x_i,y_i)$. Their total mass is denoted by $m$ and the moment of inertia about $(x_i,y_i)$ by $J$.

The kinetic energy of the system is given by
\begin{equation*}
\mathcal{K}=\frac{1}{2}\left ( J_0\dot \theta_0^2 + M(\dot x_C^2+\dot y_C^2)+J\sum_{i=1}^n\dot \theta_i^2+m\sum_{i=1}^n(\dot x_i^2+y_i^2 )\right ).
\end{equation*}
Using \eqref{E:CofM} we get
\begin{equation*}
\mathcal{K}=\frac{1}{2}\left ( (J_0+Ma^2)\dot \theta_0^2 + M(\dot x_0^2+\dot y_0^2)+2Ma\dot \theta_0(\dot y_0\cos\theta_0- \dot x_0\sin \theta_0)+J\sum_{j=1}^n\dot \theta_j^2+m\sum_{i=1}^n(\dot x_i^2+y_i^2 )\right ).
\end{equation*}

The Lagrangian of the system $\Lag:TQ\to \R$ is obtained by expressing the above quantity in terms of the 
coordinates $(x_0, y_0, \theta_0, \theta_1, \dots , \theta_n)$ of $Q$. In order to eliminate $(\dot x_i, \dot y_i)$
we note that the holonomic constraints \eqref{E:Hol-Const} imply
\begin{equation}
\label{E:Hol-Const-2}
x_i=x_0-\ell \sum_{j=1}^i\cos\theta_j, \qquad y_i=y_0-\ell \sum_{j=1}^i\sin\theta_j, \qquad i=1, \dots, n.
\end{equation}
Differentiating the above and adding yields,
\begin{equation*}
\begin{split}
\sum_{i=1}^n(\dot x_i^2 + \dot y_i^2)&=n(\dot x_0^2+\dot y_0^2)+2\ell  \sum_{j=1}^n(n+1-j)\dot \theta_j 
(\dot y_0\cos\theta_j -\dot x_0\sin \theta_j)+\ell^2 \sum_{j=1}^n(n+1-j)\dot \theta_j^2 \\
& \qquad \qquad + 2\ell^2\sum_{k=1}^{n-1}\sum_{j=k+1}^n(n+1-j)\dot \theta_k\dot \theta_j \cos(\theta_k-\theta_j),
\end{split}
\end{equation*}
where we have used the  identity
\begin{equation*}
\sum_{i=1}^n\left [ \left (\sum_{j=1}^iT_j\cos\theta_j \right)^2+ \left (\sum_{j=1}^iT_j\sin\theta_j \right)^2 \right ]
=\sum_{j=1}^n(n+1-j)T_j^2+\sum_{k=1}^{n}\sum_{j=k+1}^n2(n+1-j)T_kT_j\cos(\theta_k-\theta_j),
\end{equation*}
that holds for arbitrary scalars $T_1, \dots, T_i$.\footnote{We use the convention that a sum over an empty range of
indices equals $0$.}

Therefore, the Lagrangian of the system $\Lag:TQ\to \R$ is given by
\begin{equation}
\label{E:Lag}
\begin{split}
\Lag&=\frac{1}{2}\left ( (J_0+Ma^2)\dot \theta^2 + (M+nm)(\dot x^2+\dot y^2)+2Ma\dot \theta (\dot y\cos\theta - \dot x\sin \theta) \right . \\ & \left . \qquad +2m\ell  \sum_{j=1}^n(n+1-j)\dot \theta_j 
(\dot y\cos\theta_j -\dot x \sin \theta_j)+ \sum_{j=1}^n(J+(n+1-j)m\ell^2) \dot \theta_j^2 \right . \\
& \left . \qquad+ 2m\ell^2\sum_{k=1}^{n}\sum_{j=k+1}^n(n+1-j)\dot \theta_k\dot \theta_j \cos(\theta_k-\theta_j) \right ),
\end{split}
\end{equation}
where we have introduced the simplified notation $x=x_0, y=y_0, \theta=\theta_0$.

Using again \eqref{E:Hol-Const-2}, we can write the nonholonomic constraints \eqref{E:NonHol-Const} in 
terms of the coordinates $(x, y, \theta, \theta_1, \dots , \theta_n)$ of $Q$ as
\begin{equation}
\label{E:Constraints}
\dot x \sin \theta_i -\dot y  \cos \theta_i+\ell \sum_{j=1}^i\cos(\theta_i-\theta_j)\dot \theta_j=0, \qquad i=0, \dots ,n.
\end{equation}

In principle, using  \eqref{E:Lag} and \eqref{E:Constraints}, one could write down the equations of motion for the system in terms of Lagrange multipliers using the Lagrange-D'Alembert principle (see e.g. \cite{Landau}). However this
approach does not make use of the symmetry of the problem that we discuss next.

\subsection{Symmetries.} \label{SS:Symmetry} The system possesses an $SE(2)$ symmetry associated to the arbitrariness of the origin and
orientation of the chosen cartesian frame. The action of the matrix
\begin{equation*}
g=\left ( \begin{array}{ccc} \cos \varphi & -\sin \varphi & r \\ \sin \varphi &  \cos \varphi& s \\ 0 & 0 &1 \end{array} \right )
\in SE(2)
\end{equation*}
on the configuration $(x, y , \theta , \theta_1, \dots , \theta_n)\in Q$ is given by
\begin{equation*}
g\cdot (x , y , \theta , \theta_1, \dots , \theta_n) = (x_0\cos \varphi -y_0\sin \varphi +r,x \sin \varphi +y \cos \varphi +s,   \theta +\varphi, \theta_1+\varphi,\dots , \theta_n+\varphi).
\end{equation*}

It is immediate to check that the Lagrangian \eqref{E:Lag} and the constraints \eqref{E:Constraints}
are invariant under the tangent lift of this action. It follows that the equations of motion drop to the quotient
$\mathcal{D}/SE(2)$ which is a rank two vector bundle over the $n$-torus $\mathbb{T}^n$.

We denote the angles between subsequent bodies in the convoy by
\begin{equation}
\label{E:Defalpha}
\alpha_1=\theta- \theta_1, \qquad \alpha_{i}=\theta_{i-1}-\theta_i, \qquad i=2, \dots, n,
\end{equation}
see Figure \ref{F:Main-Diag}.
The value of these angles is invariant under the $SE(2)$ action defined above and their values serve as coordinates
on the base $\mathbb{T}^n$ of the reduced space $\mathcal{D}/SE(2)$.

Next, we denote by $u$ the component of the linear velocity of the leading body $\mathcal{B}_0$ along its main axis,
and by $\omega$ its angular velocity. We have
\begin{equation*}
u=\dot x \cos \theta + \dot y \sin \theta, \qquad \omega= \dot \theta.
\end{equation*}
As it shall become clear below, the variables $u, \omega$ serve as linear coordinates on the fibers of the 
 reduced space $\mathcal{D}/SE(2)$. The reduced equations of motion form a set of $n+2$ nonlinear, coupled, first order  ordinary differential equations for $u, \omega, \alpha_1, \dots, \alpha_n$.

\section{The equations of motion}
\label{S:EqnsMain}

 The purpose of this section is to show the following.
\begin{theorem}
\label{T:MainRed}
 The reduced equations of motion of the $n$-trailer vehicle are given by
\begin{equation}
\label{E:MainRed}
\begin{split}
\dot u &=-\frac{1}{2R(\alpha)}\left ( \sum_{k=1}^nA_k \frac{\partial R}{\partial \alpha_k} \right ) u^2 +
 \frac{Q(\alpha)}{\ell^2R(\alpha)} u\omega + \frac{Ma}{R(\alpha)}\omega^2 , \\
 \dot \omega & = - \frac{Ma u \omega }{J_0+Ma^2}, \\
 \dot \alpha_1 &=\omega - \frac{u\sin \alpha_1}{\ell}, \\
\dot \alpha_k &= \frac{u}{\ell}\left ( \prod_{j=1}^{k-2}\cos \alpha_k \right ) \left ( \sin \alpha_{k-1}- \cos \alpha_{k-1}\sin \alpha_k \right ), \qquad k=2, \dots, n.
 \end{split}
\end{equation}
where the coefficients $A_k$ are defined by \eqref{E:DefAk} below and 
\begin{equation}
\begin{split}
\label{E:defQR}
Q(\alpha)&:=\cos \alpha_1 \sin \alpha_1 \left ( m\ell^2 \sum_{j=1}^n\left ( \prod_{k=2}^j\cos^2\alpha_k \right )
-J \prod_{k=2}^n \cos^2\alpha_k \right ), \\
R(\alpha)&:=M+m\left ( \sum_{j=1}^n\prod_{k=1}^j\cos^2\alpha_k \right ) + \frac{J}{\ell^2}\left (1 - \prod_{k=1}^n\cos^2\alpha_k \right ), 
\end{split}
\end{equation}
where we denote $\alpha=(\alpha_1, \dots, \alpha_n)$.\footnote{Note that $R(\alpha)>0$ for any value of $\alpha$.}
\end{theorem}

The proof of this theorem follows the approach developed in \cite{Grab} to obtain the equations of motion  of  regular mechanical\footnote{By regular mechanical we mean that the Lagrangian is the kinetic energy minus the potential energy, where the kinetic energy defines a Riemannian metric on the configuration manifold,  and the constraint distribution 
has constant rank.}  nonholonomic system.

We begin by noting that the relations \eqref{E:Defalpha} imply
\begin{equation}
\label{E:theta-in-terms-of-alpha}
\theta_i=\theta-\sum_{j=1}^i\alpha_j, \qquad i=1,\dots, n.
\end{equation}
Using these expressions, we can write the 
nonholonomic constraints \eqref{E:Constraints} as
\begin{equation*}
\begin{split}
&\dot x \sin\theta - \dot y \cos \theta=0, \\
& \dot x \sin \left (\theta - \sum_{j=1}^i\alpha_j \right ) -\dot y \cos \left (\theta - \sum_{j=1}^i\alpha_j \right )
+ \ell \sum_{j=1}^i \cos \left ( \sum_{k=j+1}^{i}\alpha_k \right ) \left (\dot \theta - \sum_{l=1}^{j} \dot \alpha_l 
\right )=0, \quad i=1, \dots, n.
\end{split}
\end{equation*}

Use $(x,y,\theta, \alpha_1, \dots, \alpha_n)$ as coordinates on $Q$ and consider the vector fields on $Q$
\begin{equation}
\label{E:DefVF}
\begin{split}
Z_1=\cos \theta \frac{\partial}{\partial x} + \sin \theta  \frac{\partial}{\partial y}+\sum_{k=1}^n A_k\frac{\partial}{\partial \alpha_k}, \qquad
Z_2= \frac{\partial}{\partial \theta}+  \frac{\partial}{\partial \alpha_1},
\end{split}
\end{equation}
where
\begin{equation}
\label{E:DefAk}
\begin{split}
A_k&=\frac{1}{\ell}\left ( \prod_{j=1}^{k-2}\cos \alpha_j \right ) \left ( \sin \alpha_{k-1}- \cos \alpha_{k-1}\sin \alpha_k \right ),
\qquad k=1, \dots, n.
\end{split}
\end{equation}
In the above expression and in the sequel, we use the convention that the product over an empty range of indices equals $1$
and $\alpha_0=0$.

It is readily seen that $Z_1$ and $Z_2$ are linearly independent. % Moreover, using the identity
%\begin{equation*}
%\sin\left ( \sum_{j=1}^i\alpha_j\right ) = -\sum_{j=1}^i \left [ \cos \left (\sum_{k=j+1}^i\alpha_k\right )\sum_{l=1}^j\left ( \left
%( \prod_{s=1}^{l-2}\cos \alpha_s \right ) (\sin \alpha_{l-1}-\cos \alpha_{l-1} \sin \alpha_l ) \right ) \right ]
%\end{equation*}
%
Moreover, using the identities
\begin{eqnarray}
\label{E:sinId}
&&\sin\left ( \sum_{j=1}^i\alpha_j\right ) = \sum_{j=1}^i  \cos \left (\sum_{k=j+1}^i\alpha_k\right ) \left
( \prod_{s=1}^{j-1}\cos \alpha_s \right ) \sin \alpha_j,  \\
\label{E:SumProdbecomesSum}
&&\sum_{l=1}^j\left
( \prod_{s=1}^{l-2}\cos \alpha_s \right ) (\sin \alpha_{l-1}-\cos \alpha_{l-1} \sin \alpha_l ) =- \left (\prod_{s=1}^{j-1}\cos \alpha_s \right ) \sin \alpha_j,
\end{eqnarray}
one can check that $Z_1$ belongs to $\mathcal{D}$. It is easy to see that $Z_2$ is also a section
of $\mathcal{D}$. It follows that $\{Z_1, Z_2\}$ is a basis of sections of $\mathcal{D}$ and any tangent vector
 $v\in TQ$ belonging to $\mathcal{D}$ can be written as a linear combination
\begin{equation}
\label{E:sections-of-D}
v=uZ_1+\omega Z_2, \qquad u, \omega \in \R.
\end{equation}

The components of the above equation give
\begin{equation}
\label{E:Constraints-u-omega}
\begin{split}
\dot x&= u \cos \theta, \\
\dot y & = u \sin \theta, \\
\dot \theta & = \omega, \\
\dot \alpha_1 &=\omega - \frac{u\sin \alpha_1}{\ell}, \\
\dot \alpha_k &= \frac{u}{\ell}\left ( \prod_{j=1}^{k-2}\cos \alpha_j \right ) \left ( \sin \alpha_{k-1}- \cos \alpha_{k-1}\sin \alpha_k \right ), \qquad k=2, \dots, n.
\end{split}
\end{equation}

Equation \eqref{E:sections-of-D} shows that 
  $u$ and $\omega$ are linear coordinates on the fibers of $\mathcal{D}$. Moreover, the vector fields
$Z_1$ and $Z_2$ are invariant under the $SE(2)$ action defined in Section \ref{SS:Symmetry} and therefore 
they constitute a basis of sections of the reduced vector bundle $\mathcal{D}/SE(2)$. It follows that $u$ and
 $\omega$ can be interpreted as linear coordinates on the fibers of the vector bundle $\mathcal{D}/SE(2)$
  as advertised before.

Equations \eqref{E:Constraints-u-omega} are of pure kinematic nature and are well known to the control community (see e.g. \cite{Laumond1}).
They define the evolution of the variables $\alpha_1, \dots, \alpha_n$ in the reduced space and are consistent
with \eqref{E:MainRed}.

The evolution equation for $\omega$ is of dynamical nature and can be easily obtained by noting 
that the nonholonomic constraints 
as written in \eqref{E:Constraints} do not impose any restriction on the value of $\dot \theta$. Hence, the 
constraint reaction force written in the coordinates $(x, y, \theta, \theta_1, \dots, \theta_n)$ has no
component along the $\theta$-direction, and the following dynamical equation holds
\begin{equation*}
\frac{d}{dt}\left ( \frac{\partial \Lag}{\partial \dot \theta}\right ) -  \frac{\partial \Lag}{\partial \dot \theta}=0,
\end{equation*}
where $\Lag$ is given by  \eqref{E:Lag}. Explicitly we have
\begin{equation*}
(J_0+Ma^2)\ddot \theta +Ma\frac{d}{dt}( \dot y \cos\theta - \dot x \sin \theta) =-Ma\dot \theta( \dot x \cos \theta
+ \dot y \sin \theta).
\end{equation*}
Using \eqref{E:Constraints-u-omega} we obtain
\begin{equation}
\label{E:Motion-omega}
\dot \omega =- \frac{Ma u \omega }{J_0+Ma^2}
\end{equation}
as in \eqref{E:MainRed}.

The evolution equation for $u$ is more difficult to obtain. As mentioned above, we follow the approach taken
in \cite{Grab}.  This method to obtain the equations of motion of a nonholonomic system is outlined in the Appendix. 

The method requires us to compute the constrained Lagrangian $\Lag_c$ that is the restriction of $\Lag$ to $\mathcal{D}$. 
It is the kinetic energy of the system when the nonholonomic constraints are satisfied. In view of the symmetries,  
 its value can be expressed in
terms of $u, \omega, \alpha_1, \dots, \alpha_n$. To obtain an explicit expression for $\Lag_c$, start by noticing that
\eqref{E:theta-in-terms-of-alpha}, \eqref{E:SumProdbecomesSum} and \eqref{E:Constraints-u-omega} imply
\begin{equation}
\label{E:Constraints-u-omega-theta}
\dot x=u\cos \theta, \qquad \dot y = u\sin \theta, \qquad \dot \theta = \omega, \qquad
\dot \theta_k =\frac{u\sin \alpha_k}{\ell }\prod_{j=1}^{k-1}\cos \alpha_j, \quad k=1, \dots, n.
\end{equation}
Next we prove the following.

\begin{proposition}
\label{P:KineticE-ncar}
Let $j\geq 1$. If the constraints \eqref{E:Constraints-u-omega-theta} are satisfied, then we have
\begin{equation}
\label{E:KineticE-jcar}
\dot x_j^2 + \dot y_j^2 = u^2 \prod_{k=1}^j\cos^2\alpha_k .
\end{equation}
\end{proposition}
\begin{proof}
By induction. The case $j=1$ is a simple calculation using \eqref{E:Hol-Const-2} and 
\eqref{E:Constraints-u-omega-theta} and is left to the reader.
Assume that the result is valid for $j-1\geq 1$. Using  \eqref{E:Hol-Const} we write
\begin{equation*}
\dot x_j = \dot x_{j-1} + \ell \dot \theta_j \sin \theta_j, \qquad \dot y_j = \dot y_{j-1} - \ell \dot \theta_j \cos \theta_j.
\end{equation*}
Hence,
\begin{equation}
\label{E:Proof-prop1}
\dot x_j^2 + \dot y_j^2=\dot x_{j-1}^2 + \dot y_{j-1}^2 + 2\ell \dot \theta_j (\dot x_{j-1} \sin \theta_j - \dot y_{j-1} \cos \theta_j)
+ \ell^2 \dot \theta_j^2.
\end{equation}
Using \eqref{E:Hol-Const-2}  we write
\begin{equation*}
\dot x_{j-1}= \dot x + \ell \sum_{k=1}^{j-1} \sin \theta_k \dot \theta_k, \qquad 
\dot y_{j-1}= \dot y - \ell \sum_{k=1}^{j-1} \cos \theta_k \dot \theta_k,
\end{equation*}
so that 
\begin{equation*}
\dot x_{j-1} \sin \theta_j - \dot y_{j-1} \cos \theta_j=\dot x \sin \theta_j - \dot y \cos \theta_j+ \ell  \sum_{k=1}^{j-1}
\cos (\theta_j-\theta_k) \dot \theta_k.
\end{equation*}
Now, in view of \eqref{E:Constraints-u-omega-theta} and \eqref{E:theta-in-terms-of-alpha} we can write
\begin{equation*}
\begin{split}
& \dot x \sin \theta_j - \dot y \cos \theta_j=-u\sin \left ( \sum_{k=1}^j \alpha_k \right ) , \\
&\ell  \sum_{k=1}^{j-1}
\cos (\theta_j-\theta_k) \dot \theta_k =- \ell \dot \theta_j + u \sum_{k=1}^{j}\cos \left ( \sum_{l=k+1}^j \alpha_l
\right  )\sin \alpha_k  \left (\prod_{s=1}^{k-1} \cos \alpha_s  \right ).
\end{split}
\end{equation*}
Using the identity \eqref{E:sinId} we conclude that 
\begin{equation*}
\dot x_{j-1} \sin \theta_j - \dot y_{j-1} \cos \theta_j=- \ell \dot \theta_j.
\end{equation*}
Therefore, \eqref{E:Proof-prop1} becomes 
\begin{equation*}
\dot x_j^2 + \dot y_j^2=\dot x_{j-1}^2 + \dot y_{j-1}^2 - \ell^2 \dot \theta_j^2.
\end{equation*}
Using the induction hypothesis and \eqref{E:Constraints-u-omega-theta} once more, this becomes
\begin{equation*}
\dot x_j^2 + \dot y_j^2=u^2\prod_{k=1}^{j-1}\cos^2\alpha_k - u^2 \sin^2 \alpha_j \prod_{k=1}^{j-1}\cos^2\alpha_k
\end{equation*}
that is equivalent to \eqref{E:KineticE-jcar}.

%\begin{equation*}
%\begin{split}
%\dot x_1^2 + \dot y_1^2 &=u^2 - 2\ell u\sin \alpha_1 \dot \theta_1+ \ell^2  \dot \theta_1^2
%&=u^2\cos^2\alpha_1.
%\end{split}
%\end{equation*}

\end{proof}

It follows immediately from the above proposition, and from \eqref{E:Constraints-u-omega-theta}, that, if the nonholonomic constraints are satisfied, 
 the kinetic energy $\mathcal{K}_j$ of the $j^{th}$ trailer
$\mathcal{B}_j$ equals
\begin{equation*}
\begin{split}
\mathcal{K}_j&=\frac{1}{2}\left ( J\dot \theta_j^2 + m (\dot x_j^2 + \dot y_j^2) \right ) \\ &= \frac{u^2}{2}\left ( \prod_{k=1}^{j-1}
\cos^2 \alpha_k \right ) \left ( \frac{J}{\ell^2} \sin^2 \alpha_j + m \cos^2   \alpha_j  \right ),
\end{split}
\end{equation*}
for $j=1, \dots, n$.  For $j=0$ we have
\begin{equation*}
\begin{split}
\mathcal{K}_0&=\frac{1}{2}\left ( J_0\dot \theta ^2 + m (\dot x_C^2 + \dot y_C^2) \right ) \\ &= \frac{1}{2}\left ( 
(J_0+Ma^2) \omega^2 + Mu^2 \right ).
\end{split}
\end{equation*}

Therefore, adding up the contributions of all the cars in the convoy, we conclude that the constrained Lagrangian
is given by
\begin{equation}
\label{E:Const-Lag}
\Lag_c=\frac{1}{2} \left ( R(\alpha) u^2 + (J_0+Ma^2)\omega^2 \right ).
\end{equation}

%
%
%{\color{red}
%We start by computing the commutator
%\begin{equation*}
%[Z_1,Z_2]=\sin \theta \frac{\partial}{\partial x} - \cos \theta \frac{\partial}{\partial y}+\sum_{k=1}^nB_k  \frac{\partial}{\partial \alpha_k}
%\end{equation*}
%where $B_k=- \frac{\partial A_k}{\partial \alpha_1}$.
%Explicitly we have
%\begin{equation*}
%\begin{split}
%& B_1=\frac{\cos \alpha_1}{\ell}, \qquad B_2=-\frac{\cos \alpha_1 + \sin \alpha_1 \sin \alpha_2}{\ell}, \\
%& B_k=\frac{\sin \alpha_1}{\ell} \left ( \prod_{j=2}^{k-2} \cos \alpha_j \right )(\sin \alpha_{k-1} -
%\cos \alpha_{k-1}\sin \alpha_k), \quad k=3, \dots, n.
%\end{split}
%\end{equation*}
%}

Next we prove the following.
\begin{lemma}
\label{L:Proj}
 The orthogonal projection of the commutator $[Z_1,Z_2]$ onto $\mathcal{D}$ with respect to the kinetic energy metric defined
by the Lagrangian \eqref{E:Lag} is given by
\begin{equation*}
\mathcal{C}_{12}^1Z_1+\mathcal{C}_{12}^2Z_2
\end{equation*}
where
\begin{equation*}
\mathcal{C}_{12}^1=\frac{Q(\alpha)}{\ell^2R(\alpha)}, \qquad \mathcal{C}_{12}^2=-\frac{Ma}{J_0+Ma^2}.
\end{equation*}
Here $Q(\alpha)$ and $R(\alpha)$ are defined by \eqref{E:defQR}.
\end{lemma}
\begin{proof}
We give an indirect proof. In view of the discussion in the Appendix, the evolution equation for $\omega$ 
can be obtained from the general formula \eqref{E:MotionAppendix} with the subindex $b=2$ (for us $v^1=u$, $v^2=\omega$).
Since $\Lag_c$ is independent of $x,y,\theta$ and the vector field $Z_2$ is given by \eqref{E:DefVF}, 
we obtain
\begin{equation*}
\frac{d}{dt}\left ( \frac{\partial \Lag_c}{\partial \omega }\right  ) = \mathcal{C}_{12}^1u  \frac{\partial \Lag_c}{\partial u} + \mathcal{C}_{12}^2u  \frac{\partial \Lag_c}{\partial \omega} +
  \frac{\partial \Lag_c}{\partial \alpha_1},
\end{equation*}
where we have used $\mathcal{C}_{12}^e=-\mathcal{C}_{21}^e, \; e=1,2$. 

Using the expression \eqref{E:Const-Lag} for $\Lag_c$, the last equation becomes 
\begin{equation*}
(J_0+Ma^2)\dot \omega= \mathcal{C}_{12}^1R(\alpha)u^2 +  \mathcal{C}_{12}^2(J_0+Ma^2) u\omega+\frac{1}{2}
 \frac{\partial R}{\partial \alpha_1}(\alpha)u^2.
\end{equation*}
The above equation should simplify to \eqref{E:Motion-omega} so we conclude that
\begin{equation*}
 \mathcal{C}_{12}^1R(\alpha)+ \frac{1}{2}
 \frac{\partial R}{\partial \alpha_1}(\alpha)=0, \qquad  \mathcal{C}_{12}^2=-\frac{Ma}{J_0+Ma^2}.
\end{equation*}
The proof is completed by noticing that
\begin{equation}
\label{E:aux2}
\frac{\partial R}{\partial \alpha_1}(\alpha) =-\frac{2Q(\alpha)}{\ell^2}.
\end{equation}
\end{proof}

The equation for $u$ can now be obtained from the general formula \eqref{E:MotionAppendix} with the subindex $a=1$.
Since $\Lag_c$ is independent of $x,y,\theta$ and the vector field $Z_1$ is given by \eqref{E:DefVF}, 
we obtain
\begin{equation*}
\frac{d}{dt}\left ( \frac{\partial \Lag_c}{\partial u}\right  ) = -\mathcal{C}_{12}^1\omega  \frac{\partial \Lag_c}{\partial u} - 
\mathcal{C}_{12}^2\omega  \frac{\partial \Lag_c}{\partial \omega} +
 \sum_{k=1}^nA_k \frac{\partial \Lag_c}{\partial \alpha_k},
\end{equation*}
that becomes
\begin{equation}
\label{E:dotpu1}
\frac{d}{dt} ( R(\alpha) u) = -\frac{Q(\alpha)}{\ell^2}u\omega +Ma\omega^2   +\frac{1}{2}
 \sum_{k=1}^nA_k \frac{\partial R}{\partial \alpha_k} u^2.
\end{equation}

On the other hand,
\begin{equation}
\label{E:dotpu2}
\frac{d}{dt} ( R(\alpha) u)=u\sum_{k=1}^n \frac{\partial R}{\partial \alpha_k}\dot \alpha_k +R(\alpha)\dot u.
\end{equation}
Using that
\begin{equation}
\label{E:aux1}
\dot \alpha_1=\omega +u A_1 , \qquad \dot \alpha_k=uA_k, \quad k=2,\dots, n,
\end{equation}
we can combine \eqref{E:dotpu1} and \eqref{E:dotpu2} to give
\begin{equation}
\label{E:dotpu3}
R(\alpha)\dot u = -\frac{1}{2}\left (\sum_{k=1}^nA_k \frac{\partial R}{\partial \alpha_k} \right ) u^2 - \frac{\partial R}{\partial \alpha_1}u\omega - \frac{Q(\alpha)u\omega}{\ell^2}+Ma\omega^2.
\end{equation}
Using  \eqref{E:aux2} one shows that 
 equation \eqref{E:dotpu3} can be written as
\begin{equation}
\label{E:Equ}
\dot u =-\frac{1}{2R(\alpha)}\left ( \sum_{k=1}^nA_k \frac{\partial R}{\partial \alpha_k} \right ) u^2 +
 \frac{Q(\alpha)}{\ell^2R(\alpha)} u\omega + \frac{Ma}{R(\alpha)}\omega^2,
\end{equation}
that completes the proof of Theorem \ref{T:MainRed}.

\subsection{Energy conservation and the flow on the energy level surfaces}

We note that, as it is usual with nonholonomic systems, the energy is preserved. In our case, this is the reduced kinetic energy
 given by the constrained Lagrangian \eqref{E:Const-Lag}. If we define
\begin{equation}
\label{E:Energy}
\mathcal{E}(\alpha,u,\omega)= \frac{1}{2} \left ( R(\alpha) u^2 + (J_0+Ma^2)\omega^2 \right ),
\end{equation}
then a direct calculation that uses \eqref{E:aux1} and \eqref{E:aux2} shows that $\mathcal{E}$ is preserved by the flow of
equations \eqref{E:MainRed}.

Let $E>0$. It is natural to parametrize the level set $\mathcal{E}=E$ with the angles $\beta, \alpha_1, \dots, \alpha_n$ where
the angle $\beta$ is uniquely determined by the conditions
\begin{equation}
\label{E:Param-torus}
u=\sqrt{\frac{2E}{R(\alpha)}}\cos \beta, \qquad \omega=\sqrt{\frac{2E}{J_0+Ma^2}}\sin \beta.
\end{equation}
It follows that the energy level set $\mathcal{E}=E$ is diffeomorphic to the $(n+1)$-torus $\mathbb{T}^{n+1}$.
To obtain an evolution equation for $\beta$ we differentiate the above relation for $\omega$ with respect to time to obtain
\begin{equation*}
\dot \omega = \sqrt{\frac{2E}{J_0+Ma^2}}\dot \beta \cos \beta.
\end{equation*}
Now, combining \eqref{E:Motion-omega} with \eqref{E:Param-torus} and the above equation we obtain
\begin{equation*}
\sqrt{\frac{2E}{J_0+Ma^2}}\dot \beta \cos \beta=  - \frac{Ma }{J_0+Ma^2}\left ( \frac{2E}{\sqrt{R(\alpha)(J_0+Ma^2)}} \right )\cos \beta \sin \beta
\end{equation*}
which simplifies to
\begin{equation}
\label{E:betaEq}
\dot \beta =  - \frac{Ma }{J_0+Ma^2}\sqrt{\frac{2E}{R(\alpha)}}\sin \beta,
\end{equation}
assuming that $\cos \beta \neq 0$. Proceeding in an analogous fashion, differentiating the relation 
for $u$ in \eqref{E:Param-torus}
with respect to time and using  \eqref{E:Equ} we obtain   \eqref{E:betaEq}  provided that  $\sin \beta \neq 0$. In conclusion,
equation \eqref{E:betaEq} holds for any value of $\beta$.

The rest  of the equations for the flow restricted to the energy surface are obtained by combining \eqref{E:Param-torus}
with \eqref{E:MainRed}. We obtain
\begin{equation}
\label{E:alphaEqonTorus}
\begin{split}
\dot \alpha_1 &= \sqrt{\frac{2E}{J_0+Ma^2}}\sin \beta- \frac{\sqrt{2E}\sin\alpha_1}{\ell \sqrt{R(\alpha)}}\cos \beta, \\
\dot \alpha_k &=\frac{1}{\ell} \sqrt{\frac{2E}{R(\alpha)}}\left ( \prod_{j=1}^{k-2}\cos \alpha_k \right ) \left ( \sin \alpha_{k-1}- \cos \alpha_{k-1}\sin \alpha_k \right )\cos \beta, \qquad k=2, \dots, n.
\end{split}
\end{equation}

We summarize the results of this subsection in the following.
\begin{theorem}
The positive energy level sets of the reduced system \eqref{E:MainRed} are diffeomorphic to $(n+1)$-tori 
that can be parametrized with the angular variables $(\beta, \alpha_1, \dots, \alpha_n)$. The restriction of the
flow to the torus $\mathcal{E}=E>0$ is described by equations \eqref{E:betaEq} and \eqref{E:alphaEqonTorus}.
\end{theorem}

\section{Classification and linear stability of equilibria}
\label{S:dynagret0}

We study the  equilibria of the reduced system restricted to a positive energy level set. 
Throughout this section we  assume that the constant $a>0$. 

\subsection{Classification of equilibria}
\begin{proposition}
\label{P:Classification-Equil}
Let $E>0$. There exist exactly $2^{n+1}$ equlibrium points in the energy level set $\mathcal{E}=E$
of the reduced system \eqref{E:MainRed}. They are given by the conditions
\begin{equation}
\label{E:EquilConds}
u=\pm\sqrt{\frac{2E}{M+nm}},\qquad \omega=0, \qquad \sin \alpha_k=0, \quad k=1, \dots, n.
\end{equation}
\end{proposition}
\begin{proof}
We make use of the restricted equations \eqref{E:betaEq} and \eqref{E:alphaEqonTorus}. 
Imposing $\dot \beta=0$ in  \eqref{E:betaEq} implies
\begin{equation*}
\sin \beta=0.
\end{equation*}
Under this condition,  from \eqref{E:alphaEqonTorus} we see that 
we can only have $\dot \alpha_1 =0$ if $\sin \alpha_1=0$. Now assume that
$\dot \alpha_k=0$ and $\sin \beta=\sin \alpha_1=\dots = \sin \alpha_{k-1}=0$. From  \eqref{E:alphaEqonTorus}
it follows that $\sin \alpha_k=0$. This shows that the only equilibria of the system occur at the points where
\begin{equation}
\label{E:EquiCondAngles}
\sin \beta=\sin \alpha_1=\dots = \sin \alpha_{n}=0.
\end{equation}

Now use  \eqref{E:defQR} to show that the value of $R(\alpha)$ at these points is the total mass of the 
system $M+nm$. The proof is completed by using \eqref{E:Param-torus}.
\end{proof}

Assume that we are at an equilibrium configuration with energy $E$. 
The condition $\omega=0$ implies that the leading car moves along
a straight line. It moves at the constant speed $\sqrt{\frac{2E}{M+nm}}$ as indicated  by \eqref{E:EquilConds}. The motion
is forward
(in the direction from the midpoint of the wheel's axis to the center of mass) if $u>0$ or backwards if $u<0$.

On the other hand,  the condition $\sin \alpha_k=0$  in \eqref{E:EquilConds}
implies that the $k^{th}$ trailer $\mathcal{B}_k$ is aligned with  the $(k-1)^{th}$ trailer 
$\mathcal{B}_{k-1}$.  Denote by
\begin{equation}
\label{E:defsigma}
\sigma_k=\cos \alpha_k=\pm 1, \qquad k=1, \dots, n.
\end{equation}
If $\sigma_k=1$, then $\mathcal{B}_k$ is `behind' 
$\mathcal{B}_{k-1}$. If  $\sigma_{1}=-1$ then $\mathcal{B}_{1}$ `overlaps' with $\mathcal{B}_{0}$ since we assume that
the wheels of the leading car are located towards the rear of the vehicle. 
More generally, if $\sigma_{k+1}=-1$ then $\mathcal{B}_{k+1}$ `overlaps' with $\mathcal{B}_{k-1}$. 
See Figure \ref{F:Overlap}.
The situation resembles the equilibria of a chain of $n$  coupled planar pendula. 
\begin{figure}[h]
    \centering
    \includegraphics[width=18cm]{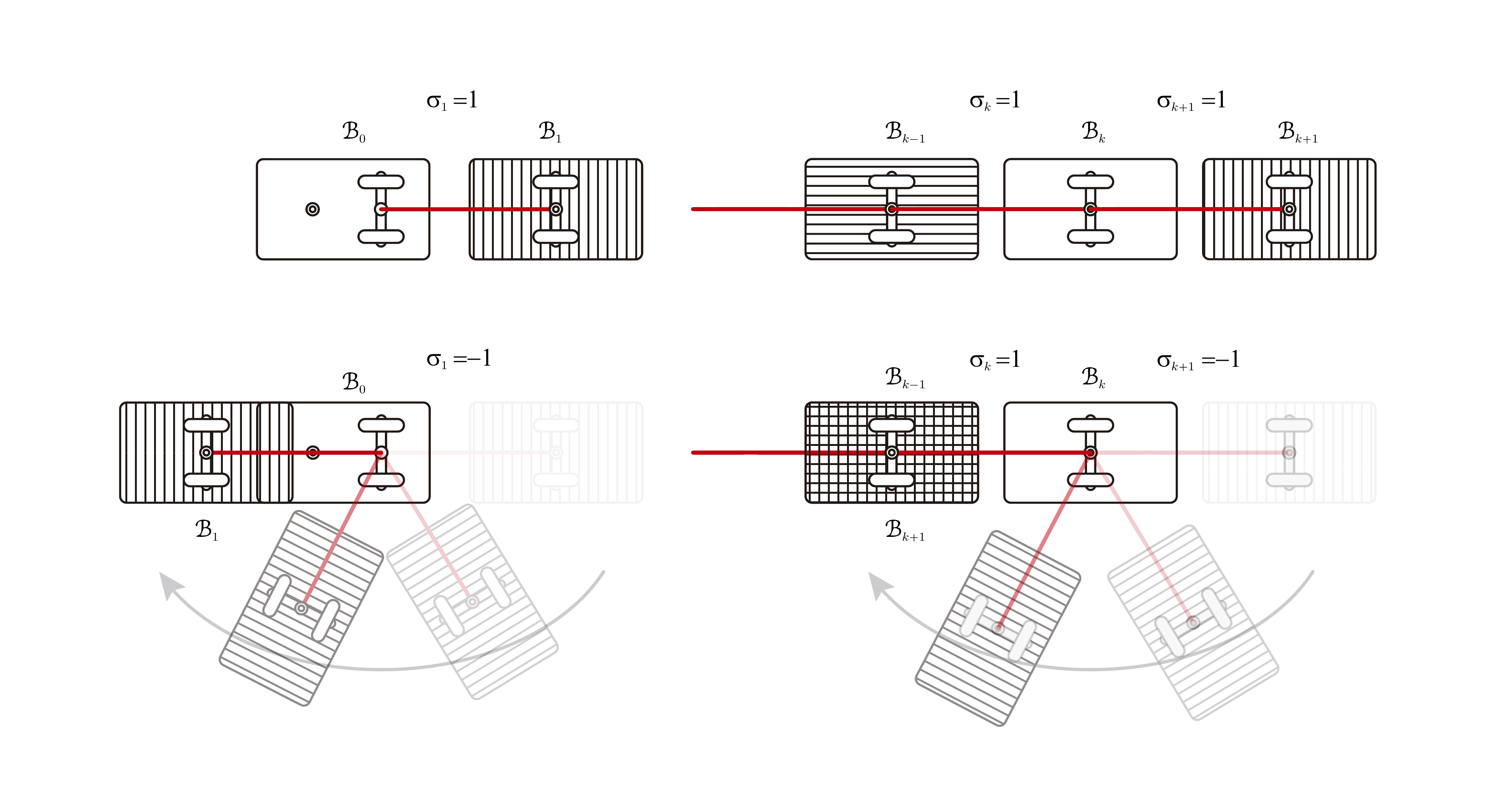}
    \caption{Illustration of equilibrium states with $\sigma_1=\pm1$ and with $\sigma_k=1$ and $\sigma_{k+1}=\pm 1$.}
    \label{F:Overlap}
\end{figure}

Therefore, the equilibria of the reduced system correspond to solutions where the convoy moves at constant speed
along a straight line
with all of the trailers aligned, with the possibility of overlaps between the cars. Of course the only 
 physically attainable equilibria occur when $\sigma_1=\cdots=\sigma_n=1$ so that there are no overlaps.
There are two of such equilibria, corresponding to forward and backward motion of the convoy. We shall see that the
former is asymptotically stable whereas the second one is asymptotically unstable.

\subsection{Stability of equilibria} We perform a linear stability analysis of the equilibria found in the previous subsection.
We will consider the system restricted to the constant energy $(n+1)$-torus $\mathcal{E}=E$, so we work
with equations \eqref{E:betaEq} and \eqref{E:alphaEqonTorus}.
To obtain the linearization of these  equations  around an equilibrium, we shall use the relations
\begin{equation*}
R(\alpha)= M+nm, \qquad \frac{\partial R}{\partial \alpha_j}(\alpha)=0, \quad j=1,\dots, n,
\end{equation*}
that hold if $\alpha=(\alpha_1, \dots, \alpha_n)$ satisfies the equilibrium conditions \eqref{E:EquilConds}.

Fix an equilibrium of equations \eqref{E:betaEq} and \eqref{E:alphaEqonTorus} satisfying \eqref{E:EquiCondAngles}. 
Denote by
\begin{equation*}
\sigma_0=\cos \beta =\pm 1.
\end{equation*}
Forward motion of the convoy corresponds to $\sigma_0=1$ and backward motion to $\sigma_0=-1$.

A straightforward calculation shows that the constant $(n+1)\times (n+1)$ matrix that defines the linearization of \eqref{E:betaEq} 
and \eqref{E:alphaEqonTorus} around the given equilibrium is 
\begin{equation*}
\frac{1}{\ell}\sqrt{\frac{2E}{M+nm}}
\left ( \begin{array}{ccccccc} \frac{-Ma \ell}{J_0+Ma^2}\sigma_0 & 0 & 0 & 0&  \cdots & 0 \\
\ell\sqrt{\frac{M+nm}{J_0+Ma^2}}\sigma_0 & -\sigma_1 \sigma_0  & 0 & 0&  \cdots & 0  \\
0 & \sigma_1\sigma_0 & - \sigma_2\sigma_1\sigma_0 & 0& \cdots & 0  \\
0 & 0 & \sigma_2 \sigma_1\sigma_0 & -\sigma_3 \sigma_2\sigma_1\sigma_0 &  \cdots & 0  \\
\vdots & \vdots & &&\ddots & \vdots \\
0 & 0 & \cdots & \cdots & & -\prod_{j=0}^n\sigma_j \end{array} \right ).
\end{equation*}
Since this matrix is lower diagonal, its eigenvalues are the diagonal components
\begin{equation*}
\lambda_0=-\frac{Ma}{J_0+Ma^2}\sqrt{\frac{2E}{M+nm}}\sigma_0, \qquad \lambda_k=-
\frac{1}{\ell}\sqrt{\frac{2E}{M+nm}}\prod_{j=0}^k\sigma_j, \quad k=1, \dots, n.
\end{equation*}
Therefore all of the equilibria are hyperbolic. Moreover, we immediately conclude the following about the nature of the equilibria.
\begin{enumerate}
\item If at least one of $\sigma_k$ with $k=1, \dots, n$, is negative (there are overlaps between the trailers) then there are
positive and negative eigenvalues and the  equilibrium is
a saddle point. 
\item If $\sigma_k=1$ for all $k=1, \dots, n$ (there are no overlaps) and $\sigma_0>0$ (the convoy is moving forwards) then
all of the eigenvalues are negative and the equilibrium is a stable node.
\item If $\sigma_k=1$ for all $k=1, \dots, n$ (there are no overlaps) and $\sigma_0<0$ (the convoy is moving backwards) then
all of the eigenvalues are positive and the equilibrium is an unstable node.
\end{enumerate}

An illustration of the numerical integration of the dynamics in the case $n=1$ is given in Figure \ref{F:anotzero}.
 Here the constant energy surface is a two-torus. It is seen the the generic initial conditions approach the stable
 (respectively unstable) node
 as $t\to \infty$ (respectively as $t\to -\infty$).
 Figure  \ref{F:anotzero} also shows  the  trajectory of the leading car $\mathcal{B}_0$ on  the plane
  for a generic initial condition. It asymptotically approaches steady motion along a straight line.
The curve traced by $\mathcal{B}_0$ closely resembles the paths followed by the Chaplygin sleigh (see e.g. \cite{NeiFu,Bo2009}).

\begin{figure}[h]
\centering
\subfigure[Phase portrait on a fundamental region of the torus $(\alpha_1, \beta)$. There are 4 equilibrium points (up to equivalence modulo $2\pi$). A stable node at $(0,0)$,  an unstable node at $(0,\pi)$ and two
saddle points at $(\pi,0)$ and $(\pi,\pi)$. ]{\includegraphics[width=9cm]{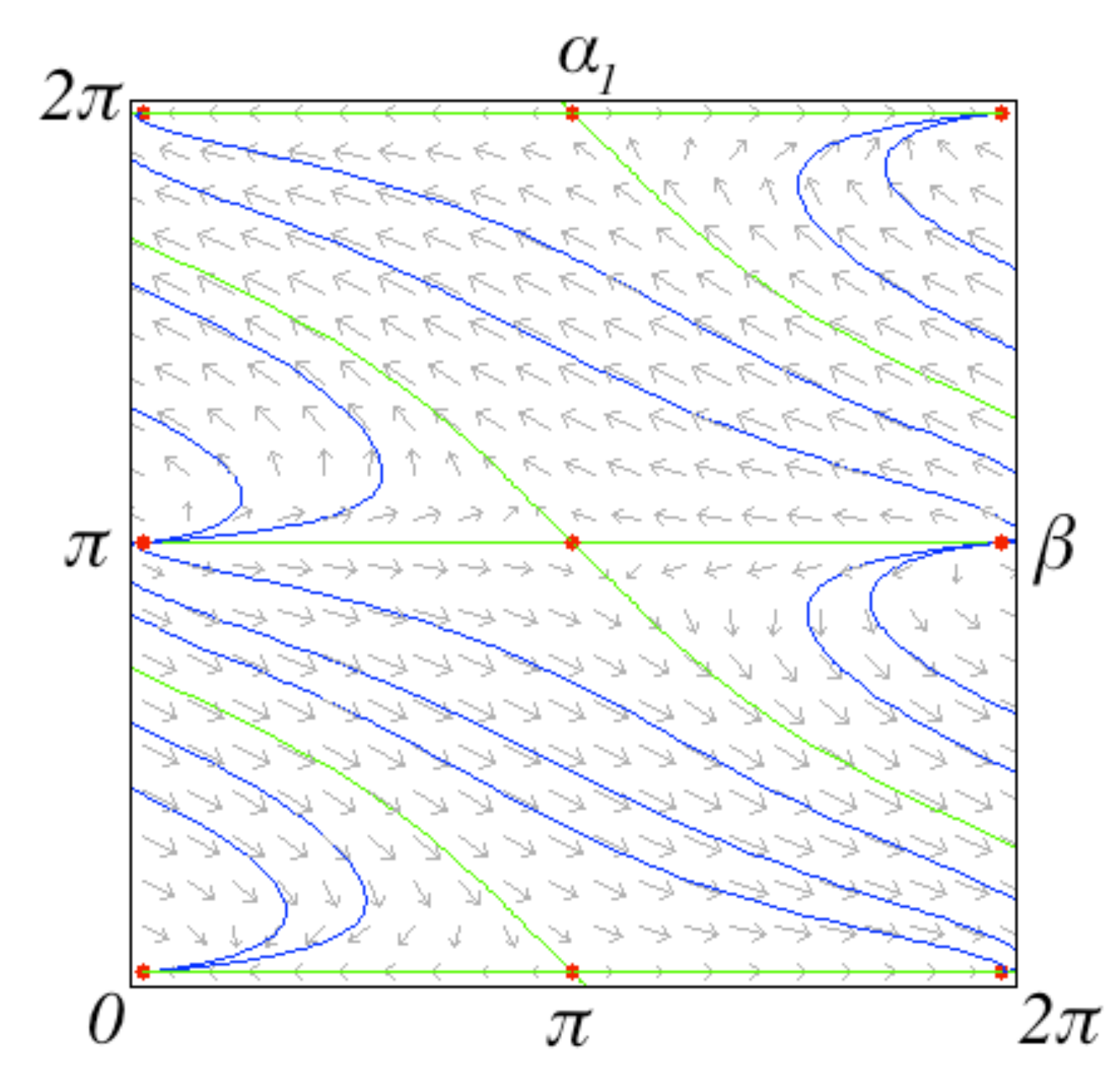}} \qquad 
\subfigure[Trajectory of  $\mathcal{B}_0$ on the plane. Asymptotic behavior towards straight line motion.]{\includegraphics[width=6cm]{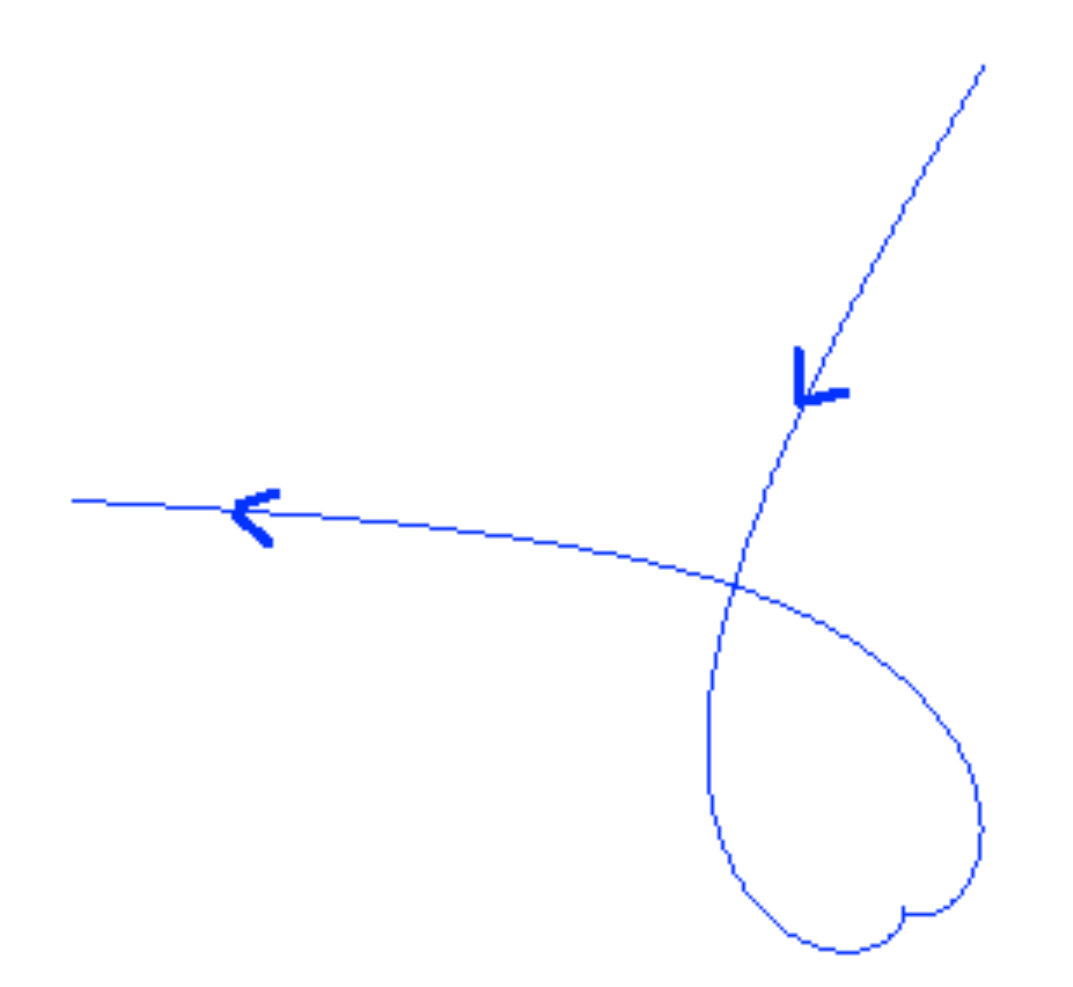}}
\caption{Phase portrait of the restriction of the reduced flow to an energy level two torus and generic trajectory of the leading car $\mathcal{B}_0$ in
the case $n=1$.} \label{F:anotzero}
\end{figure}

\section{The case $a=0$}
\label{S:dynaeq0}

If $a=0$ the dynamics changes substantially. 
From \eqref{E:MainRed} we see that $\omega$ 
is constant throughout the motion.

If $\omega=0$, the classification of the equilibrium solutions of  \eqref{E:MainRed} coincides with the
description given in Proposition \ref{P:Classification-Equil}, and the stability of the solution  with
\begin{equation*}
\sigma_0=\sigma_1=\cdots=\sigma_n=1
\end{equation*}
is analyzed in \cite{Furta}.

For the rest of the paper we consider the case where $\omega\neq 0$. The classification of equilibria is more
involved as the following proposition shows.
\begin{proposition}
\label{P:Equil-Cond-aeq0}
Suppose that $a=0$ and that $\omega=\omega_0\neq 0$. 
A necessary and sufficient condition for the existence of equilibria of \eqref{E:MainRed} 
with $u=u_0$ is that
\begin{equation}
\label{E:Equil-Cond-aeq0}
n\ell^2\omega_0^2\leq u_0^2.
\end{equation}
\end{proposition}
\begin{proof}
Equations \eqref{E:aux1} imply that at such equilibria one must have $u_0\neq 0$ and
\begin{equation}
\label{E:aux-proof-prop-aeq0}
\omega_0+ u_0A_1=0, \qquad A_k=0, \quad k=2, \dots, n.
\end{equation}
Using \eqref{E:DefAk}, the above equations can be written as
\begin{equation*}
\begin{split}
\sin \alpha_1& = \frac{\ell \omega_0}{u_0}, \\
\cos \alpha_1 \sin \alpha_2& = \frac{\ell \omega_0}{u_0} ,\\
\vdots \; & \\
\prod_{k=1}^{n-1}\cos \alpha_k \sin \alpha_n & =  \frac{\ell \omega_0}{u_0} .
\end{split}
\end{equation*}
One can inductively show that the solutions to the above equations satisfy
\begin{equation*}
\cos^2 \alpha_k = \frac{u_0^2-k\ell^2\omega_0^2}{u_0^2-(k-1)\ell^2\omega_0^2}, \qquad 
\sin^2 \alpha_k= \frac{\ell^2\omega_0^2}{u_0^2-(k-1)\ell^2\omega_0^2}, \qquad k=1,\dots, n.
\end{equation*}
It follows that a necessary condition for the existence of equilibria is  that 
\begin{equation*}
\frac{\ell^2\omega_0^2}{u_0^2-(n-1)\ell^2\omega_0^2} \leq 1,
\end{equation*}
which is equivalent to \eqref{E:Equil-Cond-aeq0}. That this condition is also sufficient is seen by noting that
if \eqref{E:aux-proof-prop-aeq0} holds (and $a=0$), the equation for $\dot u$ in \eqref{E:MainRed} becomes
\begin{equation*}
\dot u=\frac{1}{R}\left ( \frac{1}{2}\frac{\partial R}{\partial \alpha_1} +\frac{Q}{\ell^2} \right )u_0\omega_0.
\end{equation*}
But the right hand side of this equation is zero by \eqref{E:aux2}.
\end{proof}

The equations for $x, y$ and $\theta$ in \eqref{E:Constraints-u-omega} show that at an equilibrium 
solution with $u_0, \omega_0\neq 0$ the car $\mathcal{B}_0$ moves along a circle 
of radius $\frac{u_0}{\omega_0}$ at constant angular speed. Proposition \eqref{P:Equil-Cond-aeq0}
shows that the radius of this circle must be at least $\sqrt{n}\ell$.

We do not attempt to study the stability properties of the system in this case.
 Instead, we treat the
case of one trailer in detail.

\subsection{The case of one trailer.} If $a=0$ and $n=1$ then, denoting $\alpha_1=\alpha$, we have
\begin{equation*}
R(\alpha)=M+m\cos^2 \alpha +\frac{J}{\ell^2}\sin^2\alpha,
\end{equation*}
and the
 equations \eqref{E:MainRed} become
\begin{equation}
\label{E:aeq0-neq1}
\begin{split}
\dot u &=\frac{(m\ell^2 -J) u \cos\alpha \sin \alpha(\ell \omega -u\sin \alpha)}{\ell( (M+m\cos^2 \alpha)\ell^2 + J\sin^2\alpha)}\, , \\
\dot \omega & =0, \\
\dot \alpha &=\omega -\frac{u\sin \alpha}{\ell}. 
\end{split}
\end{equation}
For physical reasons it is natural to assume
\begin{equation}\label{E:Inertia-Condition}
J<m\ell^2.
\end{equation}

Equations \eqref{E:aeq0-neq1} are easily integrated using the conservation of energy. First notice that 
the level sets of the constants $E$ and 
 $\omega$ are invariant circles parametrized by $\alpha$
 \begin{equation*}
u= \pm \sqrt{\frac{2E-J_0\omega ^2}{R(\alpha)}}.
\end{equation*}

We fix a value of  $\omega=\omega_0>0$  and we study  the behavior of the flow along the invariant circle\footnote{The other cases, when either $\omega_0$ or $u$ or both are negative, are analogous.}
\begin{equation}
\label{E:Invariant-circles}
u=  \sqrt{\frac{2E-J_0\omega_0^2}{R(\alpha)}}.
\end{equation}
The evolution of $\alpha$ along the circle is given by
\begin{equation}
\label{E:Eqalpha}
\dot \alpha= \frac{\ell \sqrt{R(\alpha)}\omega_0- \sqrt{2E-J_0\omega_0^2}\sin \alpha}{\ell \sqrt{R(\alpha)}}
\end{equation}
which leads to the quadrature
\begin{equation}
\label{E:Quadrature}
\frac{\ell \sqrt{R(\alpha)} \, d\alpha}{\ell \omega_0 \sqrt{R(\alpha)} - \sqrt{2E-J_0\omega_0^2} \sin \alpha} = dt.
\end{equation}

 Now notice that 
 the inequality \eqref{E:Inertia-Condition} implies 
\begin{equation*}
M+\frac{J}{\ell^2} \leq R(\alpha) \leq M+m.
\end{equation*}
 Using \eqref{E:Invariant-circles} and the above inequality,
we see that along the solutions
of the system we have
\begin{equation*}
\begin{split}
u^2-\ell^2\omega_0^2 &= \frac{2E-J_0\omega_0^2}{R(\alpha)}- \ell^2 \omega_0^2 \\
&\leq  \frac{2E-J_0\omega_0^2}{M+\frac{J}{\ell^2}} - \ell^2 \omega_0^2 \\ &= \frac{2(E-E_c)}{M+\frac{J}{\ell^2}}
\end{split}
\end{equation*}
where
\begin{equation*}
E_c:=\frac{1}{2} \left (J_0 + J+ M\ell^2 \right) \omega_0^2.
\end{equation*}
The dynamics along the invariant circle  \eqref{E:Invariant-circles} will depend on how $E$ compares with 
$E_c$.

{\bf Case 1. If $\frac{1}{2}J_0\omega_0^2 \leq E<E_c$}. 

It follows from Proposition \ref{E:Equil-Cond-aeq0} (or directly from \eqref{E:aeq0-neq1}) that 
there are no equilibrium points of the system in this case. Hence, the dynamics along the invariant circle \eqref{E:Invariant-circles} is periodic. The energy dependent period $T=T(E)$ is obtained using \eqref{E:Quadrature}:
\begin{equation}
\label{E:period}
T=\int_0^{2\pi}\frac{\ell \sqrt{R(\alpha)} \, d\alpha}{\ell \omega_0 \sqrt{R(\alpha)} - \sqrt{2E-J_0\omega_0^2} \sin \alpha}.
\end{equation}
Using that $E<E_c$ one can verify that the denominator does not vanish so this integral is 
convergent.

{\bf Case 2. If $E=E_c$.} 

In this case there is exactly one equilibrium point along the invariant circle \eqref{E:Invariant-circles} given by 
\begin{equation*}
u=\ell \omega_0, \qquad \alpha=\frac{\pi}{2}.
\end{equation*}
Hence, the invariant circle consists of a homoclinic connection and a critical point.

{\bf Case 3. If $E>E_c$.}

We have
\begin{equation}
\label{E:Aux-Ineq}
0<\frac{\ell^2\omega_0}{2E-J_0\omega_0^2}<\frac{\ell^2\omega_0}{2E_c-J_0\omega_0^2}=
\frac{\ell^2}{J+ M \ell^2}.
\end{equation}
The graph of the function
\begin{equation*}
f(\alpha)=\frac{\sin^2\alpha}{R(\alpha)}
\end{equation*}
for $0\leq \alpha \leq \pi$ is shown in Figure \ref{F:Plotfalpha}. It is symmetrical with respect to
$\alpha=\pi/2$ where it achieves its maximum value of $\frac{\ell^2}{J+M\ell^2}$. It
attains every value between $0$ and $\frac{\ell^2}{J+M\ell^2}$ exactly two times. It follows 
from \eqref{E:Aux-Ineq} that there exist exactly two values of $\alpha$, that we denote by $\alpha^{(1)}$ and
$\alpha^{(2)}$, such that 
\begin{equation*}
0<\alpha^{(1)}<\frac{\pi}{2}<\alpha^{(2)}<\pi \quad \mbox{and} \quad  f(\alpha^{(j)})=\frac{\ell^2\omega_0}{2E-J_0\omega_0^2}, \quad j=1,2.
\end{equation*}
\begin{figure}[h]
    \centering
    \includegraphics[width=10cm]{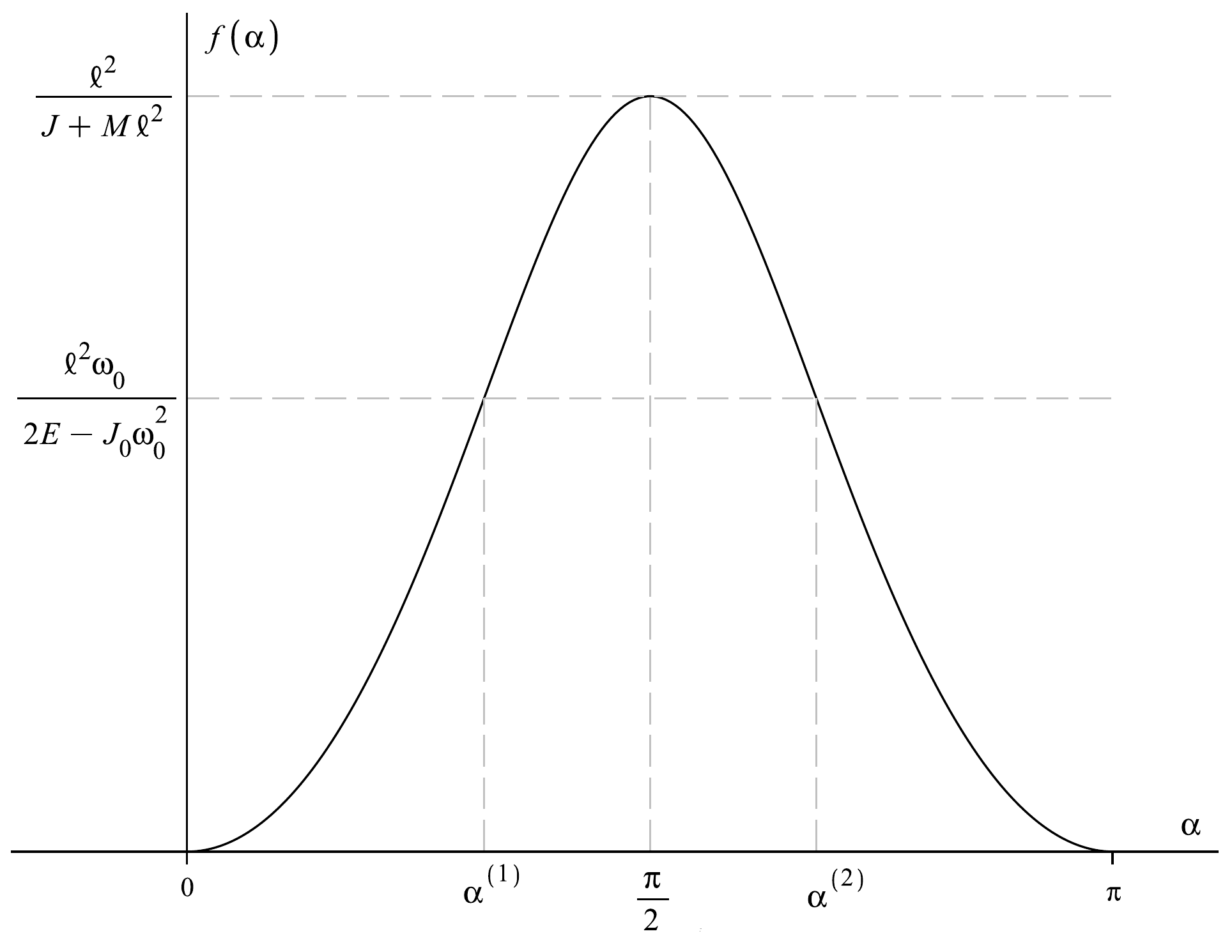}
    \caption{Plot of $f(\alpha)$. The horizontal line at height $\frac{\ell^2\omega_0}{2E-J_0\omega_0^2}$ is drawn
     under the assumption that $E>E_c$.}
    \label{F:Plotfalpha}
\end{figure}

A short calculation shows that the two points
\begin{equation*}
\alpha= \alpha^{(j)}, \qquad u= \frac{\omega_0\ell}{\sin (\alpha^{(j)})}, \qquad j=1,2
\end{equation*}
are the only equilibria of \eqref{E:aeq0-neq1} contained in the invariant circle \eqref{E:Invariant-circles}.

Given that $\sin (\alpha^{(j)})>0$, $j=1,2$, in a neighborhood of these points, we can write the evolution
equation \eqref{E:Eqalpha} for $\alpha$ as
\begin{equation*}
\dot \alpha = \omega_0-\sqrt{2E-J_0\omega_0^2}\sqrt{f(\alpha)}.
\end{equation*}
 Since $f$ is increasing at $\alpha^{(1)}$ and decreasing at $\alpha^{(2)}$ we conclude that 
 the equilibrium 
 \begin{equation}
 \label{E:Eqalpha-u}
\alpha= \alpha^{(j)}, \qquad u= \frac{\omega_0\ell}{\sin (\alpha^{(j)})}
\end{equation}
is asymptotically stable if $j=1$ and asymptotically unstable if $j=2$. A physical interpretation of these equilibria   
can be given with the aid of Figure \ref{F:Phys-Int}.
\begin{figure}[h]
    \centering
    \includegraphics[width=12cm]{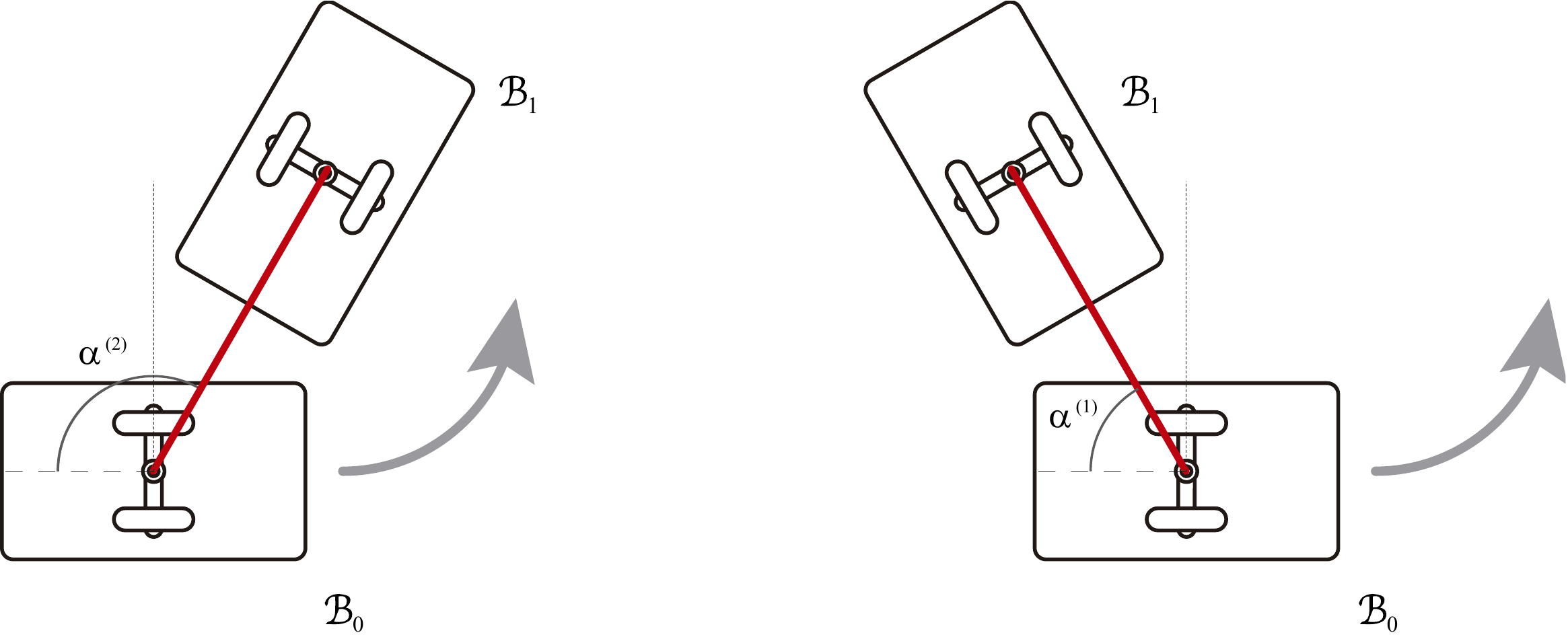}
    \caption{Illustration of the unstable and stable equilibria if $E>E_c$. The configuration on the left is unstable and on the 
    right is stable. The arrows indicate the motion of $\mathcal{B}_0$ (in agreement with our working assumption
    that $u, \omega_0>0$). }
    \label{F:Phys-Int}
\end{figure}

Hence, in this case, the invariant circle \eqref{E:Invariant-circles} consists of two heteroclinic orbits that 
connect the unstable critical point with the stable one.

Figure \ref{F:InvCircle} illustrates the qualitative dynamics on the invariant circle  \eqref{E:Invariant-circles}
in the three different energy regimes treated above. Figure \ref{F:Cylinder} illustrates the dynamics of \eqref{E:aeq0-neq1} 
on the cylinder $\omega=\omega_0>0$. 

\begin{figure}[h]
    \centering
    \includegraphics[width=15cm]{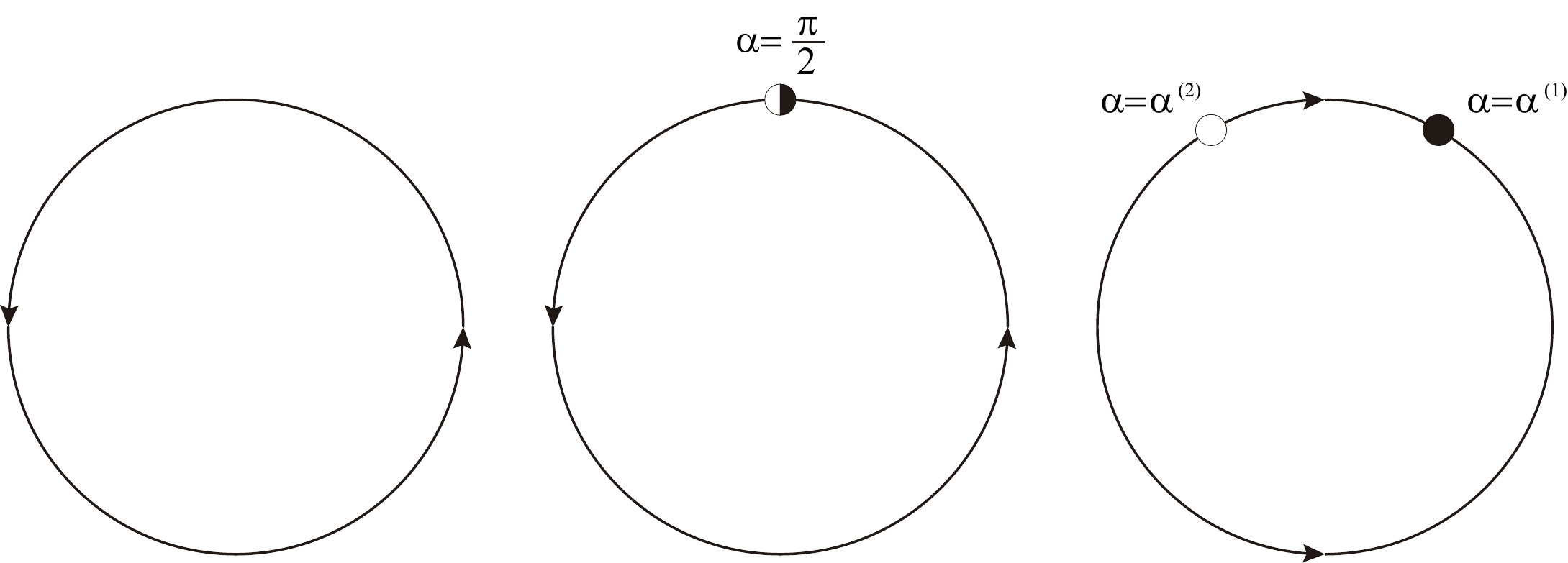}
    \caption{Qualitative behavior of the dynamics along the invariant circle  \eqref{E:Invariant-circles} for the different
    energy regimes, $E<E_c$ on the left, $E=E_c$ on the middle and   $E>E_c$ on the right.}
    \label{F:InvCircle}
\end{figure}

\begin{figure}[h]
    \centering
    \includegraphics[width=7cm]{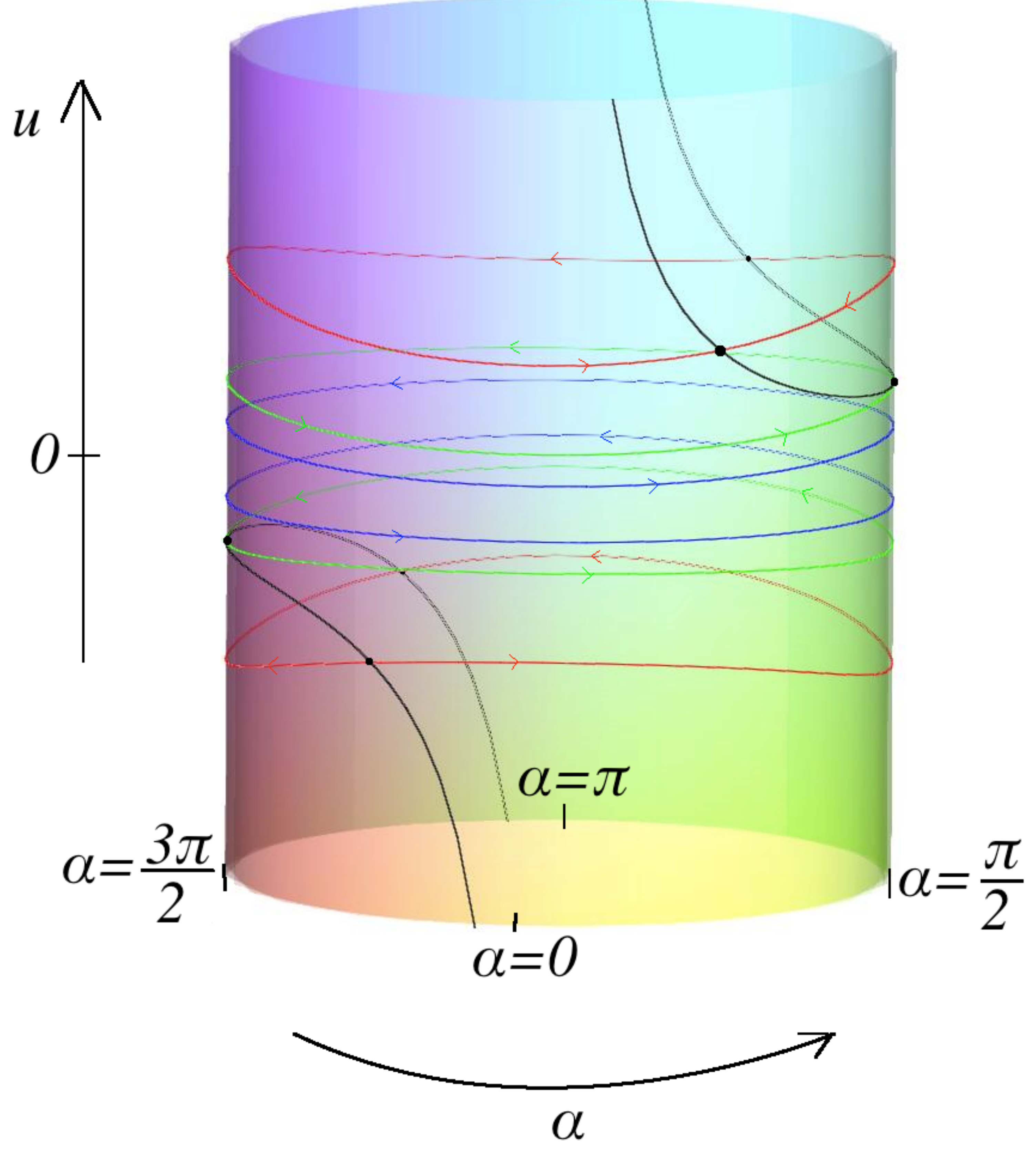}
    \caption{Dynamics on the cylinder $\omega=\omega_0>0$. The vertical axis is $u$ and the angular variable is $\alpha$. 
    It is seen that the motion takes place along the invariant circles. The critical points are indicated in black. The blue trajectories have $E<E_c$, the green ones
    $E=E_c$, and the red ones $E>E_c$. }
    \label{F:Cylinder}
\end{figure}

Our analysis shows that $E_c$ is a critical value of the energy that separates two different qualitative behaviors.
Subcritical energy values lead to periodic motion in the reduced space. On the other hand, supercritical energy values
correspond to asymptotic behavior on the reduced space. A similar phenomenon is  observed in the motion of
a Chaplygin sleigh  in a perfect fluid in the presence of circulation \cite{Fed13}.

\subsubsection{The motion on the plane} 

With the information given above, we can understand how the $2$-body convoy moves in the plane. First note
 that  in the absence of the trailer $\mathcal{B}_1$  (i.e. if $m=0$ and $J=0$) then $ u=u_0, \; 
  \omega=\omega_0$ for constants $u_0$ and $\omega_0$. Hence, the motion of $\mathcal{B}_0$ on the plane for a generic initial condition is uniform
 circular motion on a circle of radius $r=u_0/|\omega_0|$.
 
Our analysis in the previous section shows that if $E\geq E_c$ 
in the limit as $t\to \pm \infty$ the $2$-body convoy on the plane
approaches uniform circular motion. Continuing with the assumption that $u, \omega_0>0$, 
from \eqref{E:Eqalpha-u} we conclude that the radii of the limit circles is 
\begin{equation*}
r=\frac{\ell}{\sin \alpha^{(1)}}=\frac{\ell}{\sin \alpha^{(2)}}.
\end{equation*}
The value of  $\sin \alpha^{(1)}=\sin \alpha^{(2)}$ is decreasing and approaches $0$ as the energy $E\to \infty$ so the radius $r\to \infty$
for large energies. 
Figure \ref{F:asympmotion}  shows a trajectory of the leading car obtained numerically. The trailer $\mathcal{B}_1$
locks itself at a fixed angle with respect to $\mathcal{B}_0$ as $t\to \pm \infty$. The limit angles
are  $\alpha^{(2)}$ when $t\to -\infty$ and  $\alpha^{(1)}$ when $t\to \infty$. 

\begin{figure}[h]
\centering
\subfigure[$E<E_c$. Generic quasiperiodic trajectory.]{\includegraphics[width=4cm]{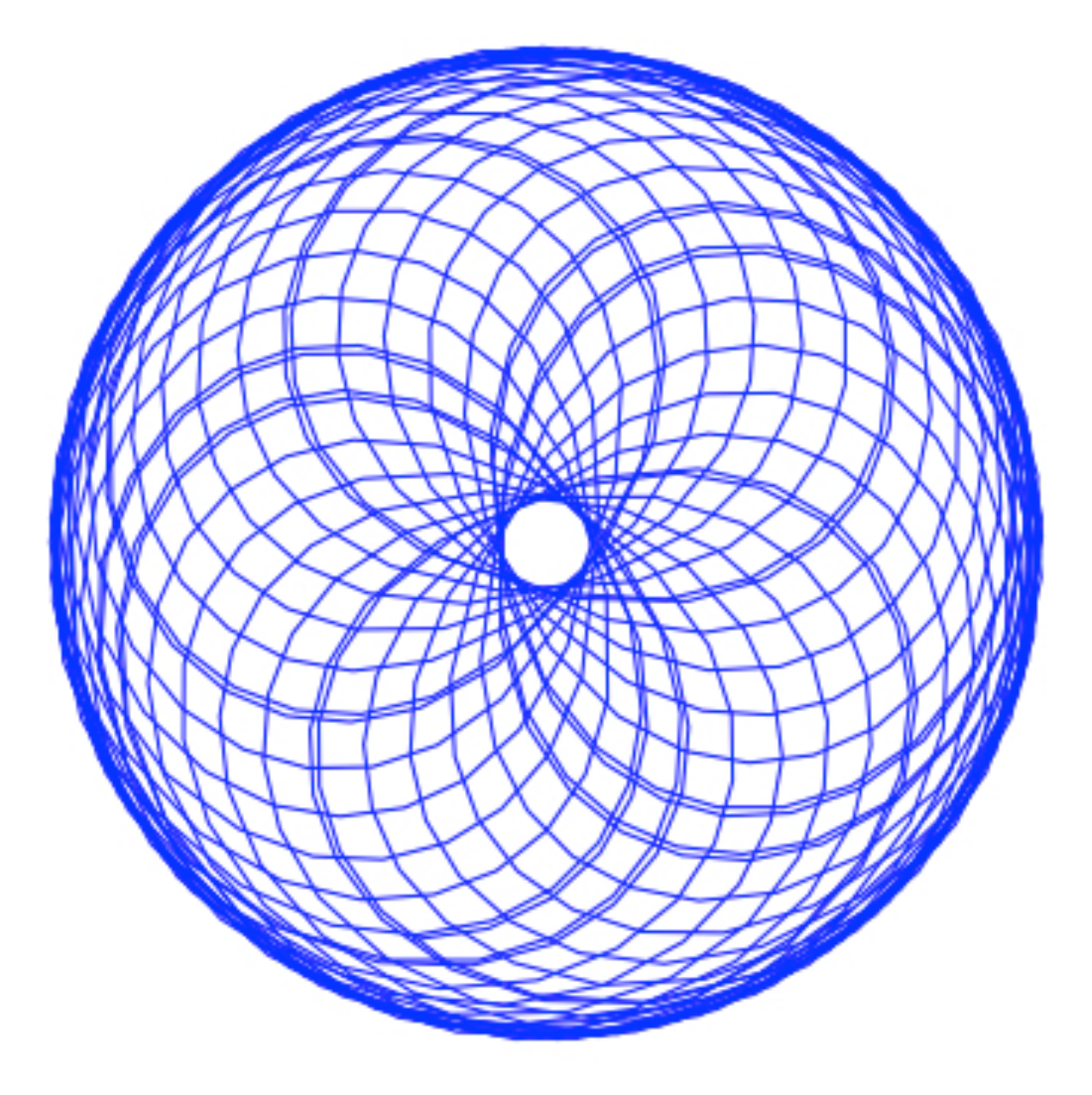}} \qquad 
\subfigure[$E<E_c$. Periodic trajectory.]{\includegraphics[width=4cm]{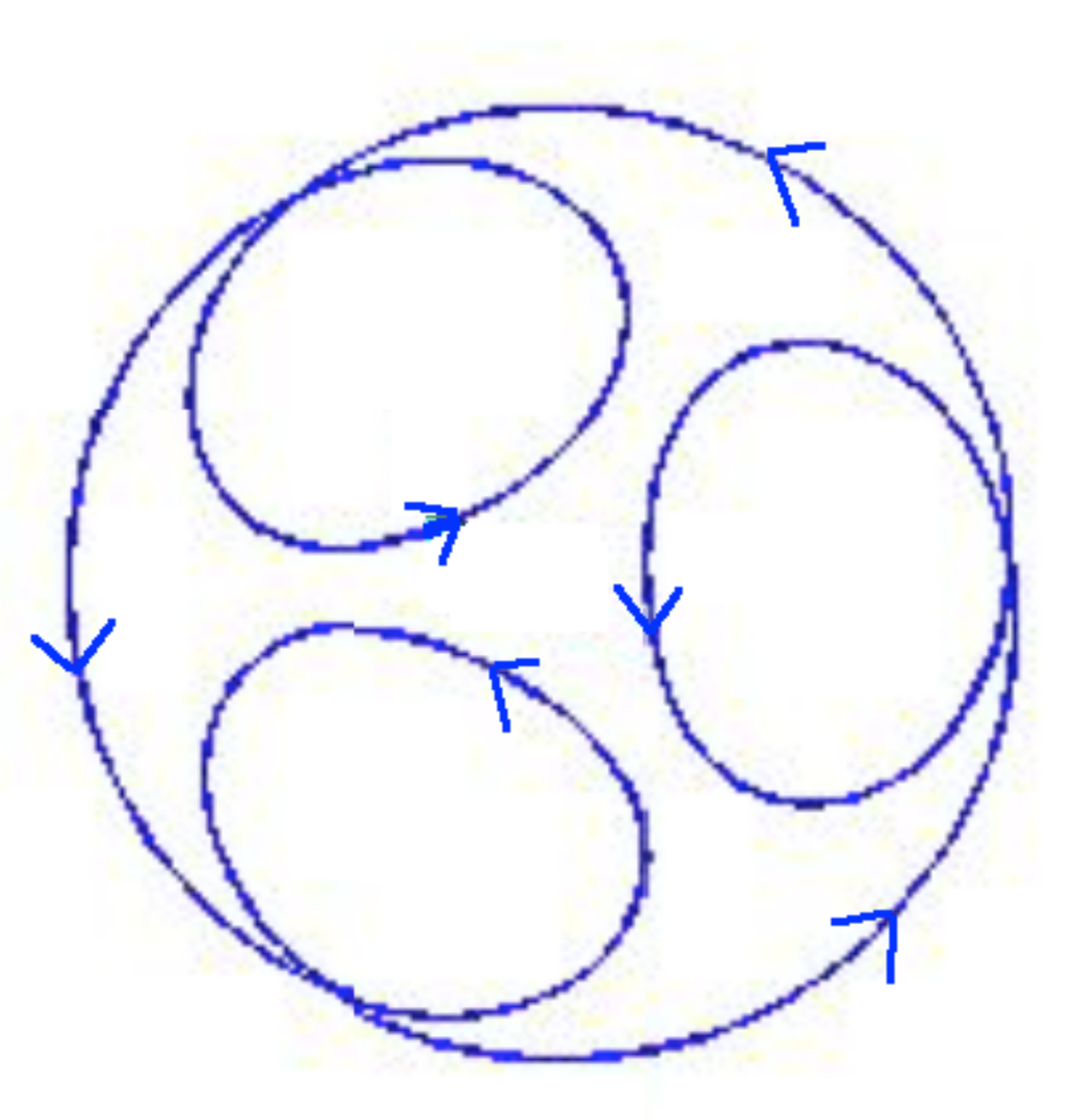}} \qquad
\subfigure[$E>E_c$. Asymptotic behavior towards limit circles.]{\includegraphics[width=4cm]{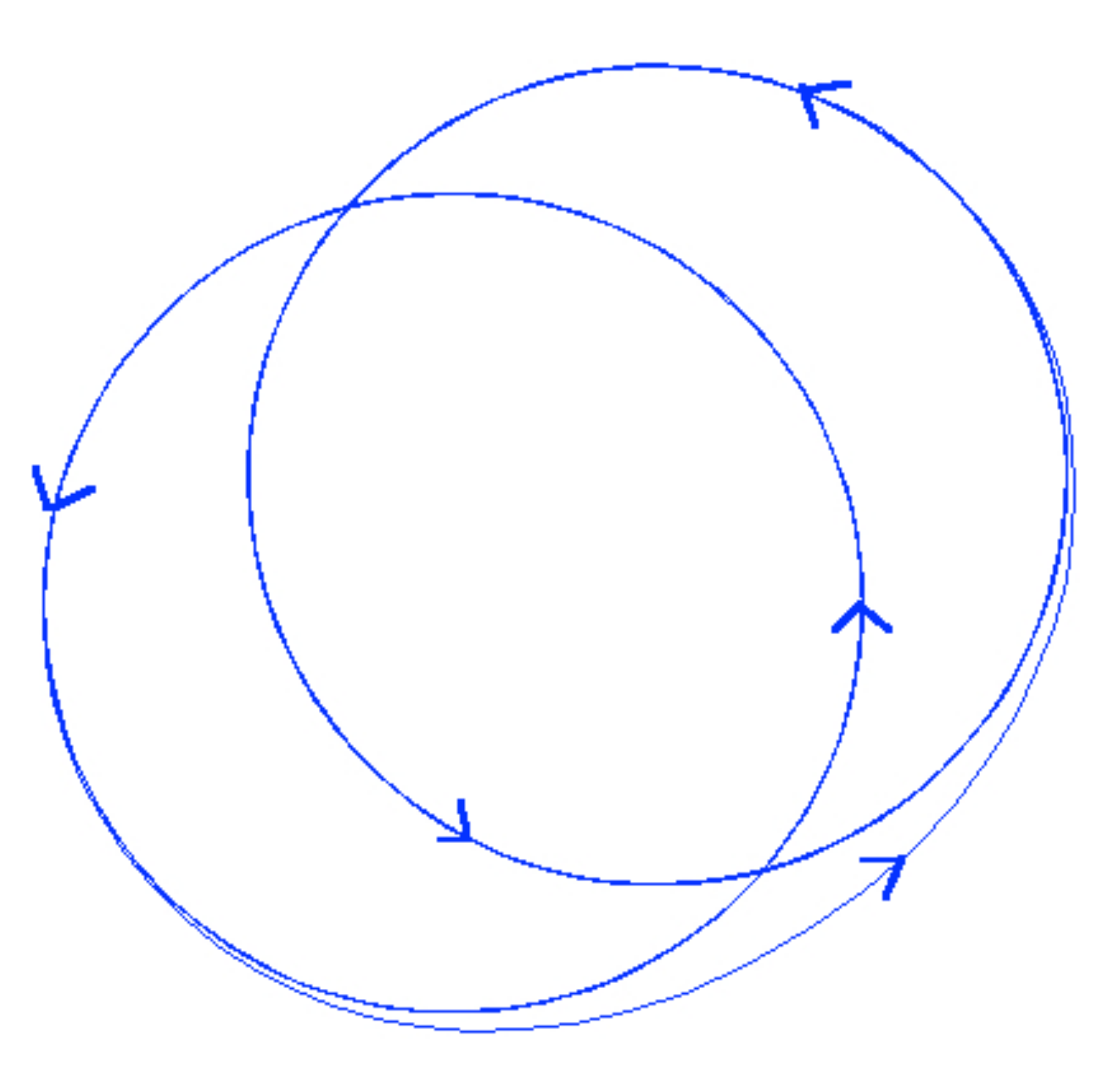}}
\caption{Trajectory of the leading car $\mathcal{B}_0$ for different energy values. } \label{F:asympmotion}
\end{figure}

On the other hand,  if $0<E< E_c$, the dynamics of $\alpha$ and $u$ is periodic with period \eqref{E:period}.
After one period, the position of the leading car $\mathcal{B}_0$ suffers a rotation by an angle $\Delta \theta = \omega_0T$,
followed by a translation by $(\Delta x, \Delta y)$ with
\begin{equation*}
\Delta x + i\Delta y=\int_0^Tu(t) e^{i\omega_0t}\, dt =\sqrt{2E-J_0\omega_0^2}\int_0^T\frac{e^{i\omega_0t}}{\sqrt{R(\alpha(t))}} \, dt,
\end{equation*}
where the dependence of $\alpha$ on $t$ is determined by \eqref{E:Quadrature} and we have assumed that $\theta(0)=0$.

Generically, the angle $\Delta \theta$ is an irrational multiple of $2\pi$ and the motion of $\mathcal{B}_0$ in the
plane is quasiperiodic  with its trajectory contained in an annulus or a circle. It is also possible 
to have periodic behavior if  $\frac{\Delta \theta}{2\pi}\in \mathbb{Q}$ or unbounded trajectories if  
$\frac{\Delta \theta}{2\pi}\in \mathbb{Z}$ and $\Delta x^2 + \Delta y^2\neq 0$. Figure \ref{F:asympmotion}  shows a periodic and a quasiperiodic trajectory
for $\mathcal{B}_0$ obtained numerically.

%\begin{equation*}
%\sigma_{k-1} \qquad  \mathcal{B}_k \qquad \mathcal{B}_{k+1}
%\end{equation*}

\section{Singular configurations}
\label{S:sing}

The {\em degree of nonholonomy}  is an important  notion that arises in 
nonlinear control theory. It expresses the level of Lie-bracketing of the elements in the constraint distribution that is needed
to span the tangent space at each configuration. This concept comes up, for instance, when trying to quantify
 the complexity associated with steering the system from one point to another (see e.g. \cite{Laumondbook}).

When the number of  trailers in our system is greater than or equal to two, this degree is not constant
throughout the configuration space.  To fix ideas we treat the case $n=2$ in detail. According to \eqref{E:DefVF}
\begin{equation*}
\begin{split}
Z_1=\cos \theta \frac{\partial}{\partial x} + \sin \theta \frac{\partial}{\partial y} -\frac{\sin \alpha_1}{\ell}  \frac{\partial}{\partial \alpha_1}+ \left ( \frac{\sin \alpha_1 - \cos\alpha_1\sin \alpha_2}{\ell} \right ) \frac{\partial}{\partial \alpha_2}, \qquad 
Z_2= \frac{\partial}{\partial \theta}+ \frac{\partial}{\partial \alpha_1},
\end{split}
\end{equation*}
form a basis of the constraint distribution $\mathcal{D}$. Direct calculations show
\begin{equation*}
\begin{split}
[Z_1,Z_2]=\sin \theta \frac{\partial}{\partial x} - \cos \theta \frac{\partial}{\partial y} +\frac{\cos \alpha_1}{\ell}  \frac{\partial}{\partial \alpha_1}- \left ( \frac{\cos \alpha_1 + \sin \alpha_1\sin \alpha_2}{\ell} \right ) \frac{\partial}{\partial \alpha_1}, \\
[Z_1,[Z_1,Z_2]]=\frac{1}{\ell^2}\frac{\partial}{\partial \alpha_1} - \left ( \frac{1+\cos \alpha_2 }{\ell^2} \right ) \frac{\partial}{\partial \alpha_2}, \qquad [Z_2,[Z_1,Z_2]]=Z_1, \\
[Z_1,[Z_1,[Z_1,Z_2]]]=\frac{\cos \alpha_1}{\ell^3}\frac{\partial}{\partial \alpha_1} - \cos \alpha_1 \left ( \frac{2+\cos \alpha_2 }{\ell^3} \right ) \frac{\partial}{\partial \alpha_2}, \qquad [Z_2,[Z_1,[Z_1,Z_2]]]=0.
\end{split}
\end{equation*}

Let $q\in Q$ be a configuration of the system with  $\cos \alpha_1\neq 0$. Then the vector fields
  $Z_1, Z_2, [Z_1,Z_2]$, $[Z_1, [Z_1,Z_2]]$ and $[Z_1,[Z_1,[Z_1,Z_2]]]$
form a basis of the tangent space $T_qQ$. The element $[Z_1,[Z_1,[Z_1,Z_2]]]$ in the basis
is said to have {\em length} 4 since one needs to compute iterated brackets of four elements in the basis of $\mathcal{D}$ to generate it. It is  clear that it is not possible to construct a basis for $T_qQ$ with iterated brackets of $Z_1$ and $Z_2$
and whose elements have length less than 4. We then say that the {\em degree of nonholonomy} at configurations
with  $\cos \alpha_1\neq 0$ is 4.

On the other hand, at configurations $q$ with $\cos \alpha_1=0$, the vector field $[Z_1,[Z_1,[Z_1,Z_2]]]$ vanishes. One can complete  $Z_1, Z_2, [Z_1,Z_2]$, $[Z_1, [Z_1,Z_2]]$  to a basis of $T_qQ$ by adjoining  the vector field
\begin{equation*}
[ [Z_1,Z_2],[Z_1,[Z_1,Z_2]]]=\frac{\sin \alpha_1}{\ell^3}\frac{\partial}{\partial \alpha_1} - \sin \alpha_1 \left ( \frac{2+\cos \alpha_2 }{\ell^3} \right ) \frac{\partial}{\partial \alpha_2},
\end{equation*}
that has length 5. Hence, the  {\em degree of nonholonomy} at configurations
with  $\cos \alpha_1\neq 0$ is 5.

The latter configurations are called singular and correspond to having $\mathcal{B}_0$ jackknifed, that is, 
 $\mathcal{B}_0$ and
$\mathcal{B}_1$ are perpendicular. It is intuitively clear that maneuvering the system at this configuration is
a more difficult task. The classification of singularities for the $n$-trailer vehicle,  and the degree of nonholonomy 
at each of them,  is given in 
\cite{Jean} for arbitrary $n$. These correspond to  different jackknifing  possibilities for  the bodies in the convoy.
A natural question is to understand what are the effects of these singular configurations on the dynamics, if any. 

 Another example of a nonholonomic system exhibiting singular configurations is  an articulated arm. 
 In recent years there have been different 
 efforts to classify the singularities of the associated  constraint distribution  \cite{Pelletier,Castro}.

To our knowledge, the effect of this kind of singularities on the motion of nonholonomic systems is
unexplored.  We hope to
report on this issue in a future note.

\renewcommand{\thesection}{\sc{Appendix}}

\section{}
%
%\subsection{Outline of the method developed in \cite{Grab} to obtain the equations of motion for a nonholonomic system}
%\label{A:EquationMotion}

\renewcommand{\thesection}{A}

The derivation of the evolution equation for $u$ in \eqref{E:MainRed} relies on the method given
in \cite{Grab} to obtain the equations of motion of a mechanical nonholonomic system. This reference includes a more
detailed description of the geometry and considers more general cases than what we need. Here we
only outline the main steps to obtain (a simple version of) their equations (3.7) and (3.8). Our presentation is
done  without proof.

Consider a nonholonomic system on a configuration manifold $Q$ of dimension $N$ with Lagrangian $\Lag:TQ\to \R$ of mechanical type and
constraint distribution $\mathcal{D}$ of constant rank $k<N$ that is bracket generating. 
The condition that $\Lag$ is of mechanical type means that it is the sum of  kinetic minus potential energy, and 
that the kinetic energy defines a Riemannian metric $\mathcal{G}$ on $Q$.

Associated to the  metric $\mathcal{G}$ there is a decomposition $TQ=\mathcal{D}\oplus \mathcal{D}^\perp$, where 
$ \mathcal{D}^\perp$ is the $\mathcal{G}$-orthogonal complement of $\mathcal{D}$, and a projection 
$\mathcal{P}:TQ\to \mathcal{D}$.

The idea is to write down the equations of motion that are consistent with the Lagrange d'Alembert principle using 
quasi-velocities that are adapted to the distribution $\mathcal{D}$. Denote by $q^1, \dots q^N$, local coordinates on 
an open set of $Q$ and by $Z_1, \dots , Z_k$ a basis of sections of $\mathcal{D}$ in such open set. That is, they
are linearly independent vector fields that lie on $\mathcal{D}$. 

Define the scalar functions  $\rho_b^i$ and $\mathcal{C}_{b d}^e$ on $Q$ through the relations\footnote{Here and in what follows we use the convention of sum over repeated indices}
\begin{equation*}
Z_b=\rho_b^i\frac{\partial}{\partial q^i}, \qquad \mathcal{P}([Z_b , Z_d])=\mathcal{C}_{bd}^e
 Z_e, \quad  b,d,e= 1, \dots k, \quad i=1, \dots, N.
\end{equation*}

Let $q\in Q$. Any tangent vector $v\in \mathcal{D}_q\subset T_qQ$  can be written
as
 \begin{equation*}
v=v^b Z_b(q)
\end{equation*}
 for certain scalars $v^b$ (the quasi-velocities). Hence, the value of the restriction of the 
 Lagrangian to $\mathcal{D}$, that we denote
 as $\Lag_c=\Lag|_\mathcal{D}$, can be expressed in terms of the variables $q^1, \dots, q^N, v^1, \dots, v^k$.
Equations (3.7) and (3.8) in \cite{Grab} state that the equations of motion for the nonholonomic system can be written as
\begin{eqnarray}
\nonumber
&&\dot q^i= \rho^i_b v^b, \qquad i=1, \dots, N, \\
\label{E:MotionAppendix}
&&\frac{d}{dt} \left ( \frac{\partial \Lag_c}{\partial v^b}\right ) =-\mathcal{C}_{bd}^ev^d \frac{\partial \Lag_c}{\partial v^e}+ \rho_b^i
 \frac{\partial \Lag_c}{\partial q^i}, \qquad b=1, \dots, k.
\end{eqnarray}
These equations avoid dealing with  Lagrange multipliers. The effect of the 
constraint forces is encoded in the effect of the projector $\mathcal{P}$ on the definition of the structure coefficients 
$\mathcal{C}_{b d}^e$.

\subsection*{Acknowledgments} We are thankful to
R. Ch\'avez-Tovar for his help to produce some of the figures, and to J.C. Marrero and A.L. Castro for useful
conversations and for indicating some references to us. LGN acknowledges
the support received from the project PAPIIT IA103815.


\begin{thebibliography}{99}


%\bibitem{Arnold} Arnold, V.I., [1989], {\it Mathematical methods of classical mechanics}. Segunda edici\'on, Graduate Texts in Mathematics {\bf 60}, Springer-Verlag.
%
%
%\bibitem{Bloch} Bloch A.~M., {\em Non-holonomic Mechanics and
%Control}. Springer Verlag, New York (2003).

\bibitem{Bolzern} Bolzern, P. DeSantis R., Locatelli, A., and Togno, S., Dynamic model of a two-trailer
articulated vehicle subject to nonholonomic constraints. {\em Robotica} {\bf 14} pp 445--450, (1996)

\bibitem{Bo2009} Borisov A. V. and Mamaev I. S. The dynamics of a Chaplygin sleigh {\em J. of Appl. Math.  Mech.} {\bf 73}
pp 156--161, (2009)

\bibitem{Lutsenko}  Borisov A.V.,  Lutsenko S.G. and  Mamaev  I.S., Dynamics of a wheeled carriage on a plane, {\em Bulletin of Udmurt University. Mathematics, Mechanics, Computer Science}, 2010, no. 4, pp. 39--48.

\bibitem{Castro} Castro A. L. and Montgomery R., Spatial curve singularities and the Monster/Semple tower, {\em
Israel J. Math.} 
{\bf 192}, 381--427, (2012).


%\bibitem{Chatila} Chatila, R., Mobile robot navigation: Space modeling and decisional processes. {\em Robotics Research: The Third International Symposium}, O Figueras and G Giralt, Eds, Cambridge MA: MIT Press, 1986,
%373--378

\bibitem{Chap} Chaplygin, S. A.,  On the theory of motion of nonholonomic systems. The theorem on the reducing multiplier. {\em Math. Sbornik} {\bf XXVIII} 303--14 (1911) (in Russian)


 \bibitem{Fed13} Fedorov Y.~N.,   Garc\'ia-Naranjo L.~C. and  Vankerschaver J., The motion of the 2D hydrodynamic Chaplygin sleigh in the presence of circulation.  {\em Disc. and Cont. Dyn. Syst. Series A} {\bf 33} (2013) no. 9, 4017--4040.

\bibitem{Furta} Fedotov A. B. and Furta S. D., On stability of motion of a chain of $n$ driven bodies. {\em Reg.  Chaot. Dyn.}
{\bf 7} 249--268, (2002). 

\bibitem{Grab} Grabowski, J., de Le\'on, M., Marrero, J. C. and Mart\'in de Diego, D. Nonholonomic constraints: a new viewpoint. {\em J. Math. Phys.} {\bf 50} (2009),  013520, 17 pp.

\bibitem{Jean} Jean, F., The car with $N$ Trailers: characterization of the singular configurations. {\em ESAIM: Control, Optimisation and Calculus of Variations}, {\bf 1}, 241--266, 1996.

\bibitem{Laumond2} Lamiraux, F., Sekhavat, S. and Laumond, J.P. Motion Planning and Control for Hilare 
Pulling a Trailer. {\em Robotics and Automatation, IEEE Transactions on}, {\bf 15} (1999), 640--652.

\bibitem{Landau} Landau L. D. and Lifshitz E. M. 1976 {\em Mechanics} 3rd edn (Oxford:Butterworth-Heinemann).

\bibitem{Laumond1} Laumond, J.P. Controllability of multibody mobile robot. {\em IEEE Trans. Robot. Automat.}, {\bf 9} (1993), 755--763.

\bibitem{Laumondbook}  Laumond J.P. {\em Robot motion planning and control}. 1998  Springer, N.ISBN 3-540-76219-1.

\bibitem{Murray} Tilbury D.,  Murray R., and  Sastry S. S., Trajectory generation for the $N$-trailer problem using Goursat normal form, Memo. UCB/ERL M93/12, Berkeley, CA, Feb. 1993.

\bibitem{Mont-Z} Montgomery R. and Zhitomirskii, M., Geometric approach to Goursat flags,
{\em Ann. I. H. Poincar\'e - AN}. {\bf 18} (2001) 459--493.

\bibitem{NeiFu} Neimark Ju I and Fufaev N A 1972 {\em Dynamics of Nonholonomic Systems (Translations of Mathematical
Monographs vol 33)} (Providence, RI: American Mathematical Society)


\bibitem{OZenk} Osborne J. and  Zenkov, D., Steering the Chaplygin sleigh using a moving
mass, {\em Proceeding on the Conference on Decision and Control (CDC-ECC)}, 2005. 

\bibitem{Pelletier} Pelletier F. and Slayman M., Configurations of an Articulated Arm
and Singularities of Special Multi-Flags. {\em SIGMA} {\bf 10}, (2014) 059, 38 pages.




\end{thebibliography}
\end{document}